\documentclass{article}

\newif\iflong\longtrue

\usepackage{authblk}
\usepackage{amssymb}
\usepackage{amsthm}
\usepackage{amsmath}
\usepackage{nicefrac}
\usepackage{dsfont}
\usepackage{mathtools}
\usepackage{graphicx}
\usepackage{hyperref}
\usepackage{cleveref}
\usepackage{tabularx}
\usepackage{enumerate}
\usepackage[top=2.5cm, bottom=2.5cm, left=2.5cm, right=2.5cm]{geometry}

\usepackage{tikz}
\usetikzlibrary{decorations.pathreplacing, decorations.pathmorphing, calc, shapes.arrows, positioning, matrix}
\usepackage{tkz-graph}
\tikzstyle{vertex}=[circle, fill, inner sep=0pt, minimum size=6pt]
\newcommand{\vertex}{\node[vertex]}
\makeatletter
\newcommand{\gettikzxy}[3]{%
  \tikz@scan@one@point\pgfutil@firstofone#1\relax
  \edef#2{\the\pgf@x}%
  \edef#3{\the\pgf@y}%
}
\makeatother

\newcommand{\proofsubparagraph}{\subparagraph}

\crefname{ineq}{inequality}{inequalities}
\creflabelformat{ineq}{#2{\upshape(#1)}#3} 
\crefname{problem}{Problem}{Problems}
\crefname{claim}{Claim}{Claims}

\newtheorem{theorem}{Theorem}
\theoremstyle{definition}

\theoremstyle{definition}
\newtheorem{definition}{Definition}
\theoremstyle{remark}

\theoremstyle{plain}

\newtheorem{corollary}[theorem]{Corollary}
\newtheorem{lemma}[theorem]{Lemma}
\newtheorem*{lemma*}{Lemma}
\newtheorem*{claim}{Claim}

\newcommand*{\myproofname}{Proof}
\newenvironment{claimproof}[1][\myproofname]{\begin{proof}[#1]}{\end{proof}}

\newcommand{\poly}{\operatorname{poly}}
\newcommand{\polylog}{\operatorname{polylog}}

\newcommand{\cm}{\operatorname{cm}}

\newcommand{\nae}{\operatorname{NAE}}

\def\X{\ensuremath\mathcal{X}}
\def\R{\ensuremath\mathcal{R}}
\def\true{{\texttt{true}}}
\def\false{{\texttt{false}}}

\newcommand{\PLS}{{\normalfont\sffamily{PLS}}}
\newcommand{\PSPACE}{{\normalfont\sffamily{PSPACE}}}

\newcommand{\NP}{{\normalfont\sffamily{NP}}}

\newcommand{\hartigan}{{\textsc{2-Means/Flip}}}
\newcommand{\khartigan}{{\textsc{$k$-Means/Flip}}}
\newcommand{\kkhartigan}{{\textsc{$(k+1)$-Means/Flip}}}

\newcommand{\emaxcut}{{\textsc{Squared Euclidean Max Cut/Flip}}}
\newcommand{\emaxcutnosq}{{\textsc{Euclidean Max Cut/Flip}}}

\newcommand{\maxcut}{{\textsc{Max Cut/Flip}}}
\newcommand{\maxcutd}[1][]{{\textsc{Max Cut-$#1$/Flip}}}
\newcommand{\maxcutdistinctd}[1][]{{\textsc{Distinct Max Cut-$#1$/Flip}}}
\newcommand{\minbisection}{{\textsc{Odd Min Bisection/Flip}}}
\newcommand{\maxbisection}{{\textsc{Odd Max Bisection/Flip}}}

\newcommand{\posnaesat}[1][]{\textsc{Pos NAE $#1$-SAT/Flip}}
\newcommand{\halfposnaesat}[1][]{\textsc{Odd Half Pos NAE $#1$-SAT/Flip}}
\newcommand{\densestcut}{{\textsc{Densest Cut/Flip}}}
\newcommand{\sparsestcut}{{\textsc{Sparsest Cut/Flip}}}

\newcommand{\flip}{{\textsc{Flip}}}

\newcommand{\problemdef}[3]{
    \vspace{-3pt}
    {\centering
    \bigskip\par\noindent%
    \renewcommand{\arraystretch}{1.2}%
      \begin{tabularx}{0.9\textwidth}{@{}l@{\hspace{3pt}}X}
        \hline\hline%
        \multicolumn{2}{@{}l}{\normalsize\textsc{#1}} \\[\fboxsep]%
        \normalsize\textbf{Input:}    & \normalsize#2 \\
        \normalsize\textbf{Output:} & \normalsize#3\\\hline\hline%
      \end{tabularx}
    \medskip\par
    \vspace{2pt}}%
}

\let\originalleft\left
\let\originalright\right
\renewcommand{\left}{\mathopen{}\mathclose\bgroup\originalleft}
\renewcommand{\right}{\aftergroup\egroup\originalright}

\title{Complexity of Local Search for Euclidean Clustering Problem}

\author[1]{Bodo Manthey}
\author[2]{Nils Morawietz}
\author[1]{Jesse van Rhijn}
\author[3]{Frank Sommer}
\affil[1]{Faculty of Electrical Engineering, Mathematics, and Computer Science, University of Twente, The Netherlands}
\affil[2]{Friedrich Schiller University Jena, Institute of Computer Science, Germany}
\affil[3]{Institute of Logic and Computation, TU Wien, Vienna, Austria}
\affil[ ]{\textit {\{b.manthey, j.vanrhijn\}@utwente.nl, nils.morawietz@uni-jena.de, fsommer@ac.tuwien.ac.at}}

\date{}

\begin{document}
\maketitle

\begin{abstract}
    We show that the simplest local search heuristics for two natural Euclidean clustering
    problems are \PLS{}-\iflong complete\fi.
    First, we show that the Hartigan--Wong method, which is essentially the \textsc{Flip} heuristic, for \textsc{$k$-Means} clustering is \PLS{}-\iflong complete\fi,
    even when $k = 2$.
    Second, we show the same result for the \textsc{Flip} heuristic for \textsc{Max Cut}, even
    when the edge weights are given by the (squared) Euclidean distances between
    the points in some set $\X \subseteq \mathds{R}^d$; a problem
    which is equivalent to \textsc{Min Sum 2-Clustering}. 
\end{abstract}

\section{Introduction}

Clustering problems arise frequently in various fields of application.
In these problems, one is given a set of objects, often represented
as points in $\mathds{R}^d$, and is asked to partition the set into
\emph{clusters}, such that the objects within a cluster
are similar to one another by some measure. For points in $\mathds{R}^d$, a natural
measure is the (squared) Euclidean distance between two objects.
In this paper, we consider two Euclidean clustering problems that use this similarity
measure: \textsc{$k$-Means} clustering and \textsc{Squared Euclidean Max Cut}. 

\subparagraph*{\boldmath \texorpdfstring{$k$}{k}-Means.}

One well-studied clustering problem is \textsc{$k$-Means}
\cite{berkhinSurveyClusteringData2006,jainDataClustering502010}.
In this problem, one is given a set of points $\X \subseteq \mathds{R}^d$ and an integer
$k$. The goal is to partition $\X$ into exactly $k$ clusters such that
the total squared distance of each point to the centroid of its cluster
is minimized. Formally, one seeks to minimize the clustering cost

$$
    \sum_{i=1}^k \sum_{x \in C_i} \|x - \cm(C_i)\|^2 \quad\text{where}\quad \cm(C_i)
        = \frac{1}{|C_i|} \sum_{x \in C_i} x.
$$

Being \NP-hard even when $k = 2$ \cite{ADHP09} or when $\X \subseteq \mathds{R}^2$
\cite{mahajanPlanarKmeansProblem2012},
\textsc{$k$-Means} has been extensively studied from the perspective of approximation
algorithms\iflong, with several such approaches having emerged\fi    \cite{arthurKmeansAdvantagesCareful2007,hasegawaEfficientAlgorithmsVarianceBased2000,kanungoEfficientKmeansClustering2002,matousekApproximateGeometricClustering2000}. 
Nevertheless, local search remains the method of choice for practitioners
\cite{berkhinSurveyClusteringData2006,jainDataClustering502010}.

The most well-known local search algorithm for \textsc{$k$-Means} is
Lloyd's method \cite{lloydLeastSquaresQuantization1982}.
Here, one alternates between two steps in each iteration.
In the first step, each point is assigned to its closest cluster center, and in the second step the
cluster centers are recalculated from the newly formed clusters.

This algorithm was shown to have worst-case super-polynomial running time by
Arthur and Vassilvitskii \cite{arthurHowSlowKmeans2006},
with the result later improved to exponential running time
even in the plane by Vattani \cite{vattaniKmeansRequiresExponentially2011}.
Moreover, Roughgarden and Wang \cite{roughgardenComplexityKmeansMethod2016}
showed it can implicitly solve \PSPACE-complete problems. 
On the other hand, Arthur et al.\ \cite{arthurSmoothedAnalysisKMeans2011} proved that Lloyd's method has
smoothed polynomial running time on Gaussian-perturbed point sets,
providing a degree of explanation for its effectiveness in practice.

Recently Telgarsky and Vattani~\cite{telgarskyHartiganMethodKmeans2010}
revived interest in another, older local search
method for \textsc{$k$-Means} due to Hartigan and Wong \cite{hartiganAlgorithm136KMeans1979}.
This algorithm, the Hartigan--Wong method, is much simpler:
one searches for a single point that can be reassigned to some other cluster
for a strict improvement in the clustering cost. 
In other words, the Hartigan--Wong method is the \textsc{Flip} heuristic.
In the following, we always use \textsc{Flip} instead of Hartigan--Wong to indicate that this is the most simple heuristic for this problem and to keep the name for the used heuristic consistent.
Despite this simplicity,
Telgarsky and Vattani~\cite{telgarskyHartiganMethodKmeans2010} show that the \textsc{Flip} heuristic is more powerful than Lloyd's method, in the sense
that the former can sometimes improve clusterings produced by the latter, while the converse
does not hold.

A similar construction to that of Vattani for Lloyd's method
shows that there exist instances on which the \textsc{Flip} heuristic
can take exponentially many iterations to find a local optimum,
even when all points lie on a line \cite{mantheyWorstCaseSmoothedAnalysis2024b}.
However, this example follows a contrived sequence of iterations. Moreover,
\textsc{$k$-Means} can be solved optimally for instances in which all points lie on a line.
Thus, the question
remains whether stronger worst-case examples exist, and what the complexity
of finding locally optimal clusterings is.

\subparagraph*{Squared Euclidean Max Cut.}

Another clustering problem similar to \textsc{$k$-Means} is
\textsc{Squared Euclidean Max Cut}.
Recall that \textsc{Max Cut} asks for a subset
of vertices $S$ of a weighted graph $G = (V, E)$, such that the
total weight of the edges with one endpoint in $S$ and one in~$V \setminus S$
is maximized. This problem emerges in numerous applications, from graph clustering
to circuit design to statistical physics~\cite{barahonaApplicationCombinatorialOptimization1988,boykovInteractiveGraphCuts2001}.

In \textsc{Squared Euclidean Max Cut}, one identifies
the vertices of $G$ with a set $\X \subseteq \mathds{R}^d$, and assigns
each edge a weight equal to the squared Euclidean distance between its endpoints.
This problem is equivalent to \textsc{Min Sum 2-Clustering} (although not in approximation),
where one seeks to minimize

$$
    \sum_{x, y \in X} \|x - y\|^2 + \sum_{x, y \in Y} \|x - y\|^2
$$
over all partitions $(X, Y)$ of $\X$.
Also this special case of \textsc{Max Cut} is \NP-hard \cite{ageevComplexityWeightedMaxcut2014}.
In a clustering context, the problem was studied by 
Schulman \cite{schulmanClusteringEdgecostMinimization2000}
and Hasegawa et al.\ \cite{hasegawaEfficientAlgorithmsVarianceBased2000},
leading to exact and approximation algorithms.

Given the computational hardness of \textsc{Max Cut}, practitioners
often turn to heuristics.
Some of the resulting algorithms are
very successful, such as the Kernighan-Lin heuristic \cite{kernighanEfficientHeuristicProcedure1970}
and the Fiduccia-Mattheyses algorithm \cite{fiducciaLinearTimeHeuristicImproving1982a}.
Johnson et al.\ \cite{johnsonHowEasyLocal1988}
note that the simple \textsc{Flip} heuristic, where one moves a single vertex
from one side of the cut to the other, tends to converge quickly in practice
.
Sch\"affer and Yannakakis \cite{schafferSimpleLocalSearch1991} later showed
that it has exponential running time in the worst case.

One may wonder whether \textsc{Flip} performs
better for \textsc{Squared Euclidean Max Cut}. Etscheid and R\"oglin
\cite{etscheidSmoothedAnalysisSquared2015a,etscheidWorstCaseAnalysisMaxCut2018}
performed a smoothed analysis of \textsc{Flip} in this context, showing a smoothed
running time of $2^{O(d)}\cdot\poly(n, 1/\sigma)$ for Gaussian-perturbed instances,
where $\sigma$ denotes the standard deviation of the Gaussian noise.
On the other hand, they also exhibited an instance in $\mathds{R}^2$
on which there exists an exponential-length improving sequence of iterations,
with the caveat that not all edges are present in the instance. Like for
\textsc{$k$-Means}, one may ask whether stronger examples exist (e.g.\ on complete
graphs), and what the complexity of finding \flip-optimal solutions is.

\subparagraph*{Complexity of Local Search.}

The existence of instances with worst-case exponential running time is
common for local search heuristics. 
To investigate this phenomenon, and local search heuristics in general,
Johnson et al.~\cite{johnsonHowEasyLocal1988} defined a complexity class
\PLS{}, for polynomial local search. 
The class is designed to capture the properties
of commonly used local search heuristics and contains pairs consisting of an optimization 
problem \textsc{P} and a local search heuristic \textsc{$\mathcal{N}$}.
In the following we denote such a pair as \textsc{P/$\mathcal{N}$}. 
\PLS{}-complete problems have~the property
that their natural local search algorithms have worst-case exponential running time.
\iflong In the same work, Johnson et al. \fi showed that the Kernighan-Lin heuristic for
the \textsc{Max Bisection} problem (a variant of \textsc{Max Cut}, where the parts of the partition must be of equal size)
is \PLS{}-complete. 
This stands
in contrast to the empirical success of this algorithm~\cite{kernighanEfficientHeuristicProcedure1970}.

Building on this work, Sch\"affer and Yannakakis \cite{schafferSimpleLocalSearch1991}
later proved that a host of very simple local search heuristics are \PLS{}-complete,
including the \textsc{Flip} heuristic for \textsc{Max Cut}. This refuted a conjecture by
Johnson et al., who doubted that such simple heuristics could be \PLS{}-complete.
Els\"asser and Tscheuschner \cite{elsasserSettlingComplexityLocal2011}
later showed that this remains true even in the very
restricted variant where the input graph has maximum degree five, which we will
refer to as \textsc{Max Cut-5}.

Sch\"affer and Yannakakis defined a new type of \PLS{}-reduction
called a \emph{tight} reduction. In addition to showing completeness for \PLS{},
this type of reduction also transfers stronger properties on the running time
of local search heuristics between \PLS{} problems.

Since the introduction of \PLS{}, many local search problems have been shown to be
\PLS{}-complete, including such successful heuristics as Lin-Kernighan's algorithm for the
TSP \cite{papadimitriouComplexityLinKernighan1992} or the~\textsc{$k$-Swap}-neighborhood heuristic for~\textsc{Weighted Independent Set}~\cite{komusiewiczMorawietzFindingSwap2022} for~$k\geq 3$. 
For a non-exhaustive list,
see Michiels, Korst and Aarts \cite[Appendix C]{michielsTheoreticalAspectsLocal2007}.

\subparagraph*{Our Contribution.}

Given the existence of \textsc{$k$-Means} instances where the
\textsc{Flip} heuristic has worst-case exponential running time, one may
ask whether this heuristic is \PLS{}-\iflong complete\fi. In this work,
we answer this question in the affirmative.
\begin{theorem}\label{thm:hartigan}
    For each~$k\ge 2$, \khartigan{} is \PLS{}-\iflong complete\fi.
\end{theorem}

Just as with \khartigan{}, we ask whether \textsc{Squared Euclidean Max Cut}
with the \textsc{Flip} heuristic is \PLS{}-\iflong complete\fi. Again, we answer this question
affirmatively. In addition, we show the same result for
\textsc{Euclidean Max Cut}, where the distances between the points are not squared.
\begin{theorem}\label{thm:maxcut}
    \emaxcutnosq{} and \emaxcut{} are \PLS{}-\iflong complete\fi.
\end{theorem}

We note that \PLS{}-\iflong completeness \fi results for Euclidean local optimization problems
are rather uncommon. We are only aware of one earlier result by
Brauer \cite{brauerComplexitySingleSwapHeuristics2017},
who proved \PLS{}-completeness of a local search heuristic for
a discrete variant of \textsc{$k$-Means}. This variant chooses $k$ cluster
centers among the set of input points, after which points are assigned to their
closest center. The heuristic they consider removes one point from the set
of cluster centers and adds another. In their approach, they first construct a metric
instance, and then show that this instance can be embedded into $\mathds{R}^d$ using
a theorem by Schoenberg \cite{schoenbergMetricSpacesPositive1938}. In contrast, we
directly construct instances in $\mathds{R}^d$.

In addition to showing \PLS{}-hardness of \khartigan{} and \emaxcut{}, we also
show that there exist specific hard instances of these problems, as well as all other problems
considered in this paper.

\begin{theorem}\label{thm:hardinstances}
    Each local search problem~$L$ considered in this work (see~\Cref{sec:definitions}) fulfills the following properties:
    \begin{enumerate}
     \item $L$ is \PLS-\iflong complete\fi.
     \item 
     For each~$n$, one can compute in polynomial time an instance of~$L$ of size~$n^{\mathcal{O}(1)}$ with an initial solution that is exponentially far away from any local optimum.
    \item The problem of computing the locally optimal solution obtained from performing a standard
    local search algorithm based on the neighborhood of $L$
    is \PSPACE{}-\iflong complete \fi for $L$.\label{PSCPACE prop}
    \end{enumerate}
\end{theorem}

A formal definition of Property~\ref{PSCPACE prop} is given in~\Cref{section PLS}. 
In particular, this result shows that there exists an instance of \textsc{$k$-Means}
such that there exists an initial solution to this instance that is exponentially
many iterations away from \emph{all} local optima~\cite{schafferSimpleLocalSearch1991}.
By contrast, the earlier result \cite{mantheyWorstCaseSmoothedAnalysis2024b}
only exhibits an instance with a starting solution that is exponentially
many iterations away from \emph{some} local optimum. Moreover, \Cref{thm:hardinstances}
yields instances where \textsc{Flip} has exponential running time in \textsc{Squared Euclidean Max Cut}
on \emph{complete} geometric graphs, unlike the older construction which omitted some edges
\cite{etscheidWorstCaseAnalysisMaxCut2018}.

\section{Preliminaries and Notation}

Throughout this paper, we will consider all graphs to be undirected unless
explicitly stated otherwise.
Let $G = (V, E)$ be a graph. For $v \in V$, we denote
by $d(v)$ the \emph{degree} of $v$ in $G$, and by $N(v)$ the set of \emph{neighbors} of $v$. 

Let $S, T \subseteq V$. We write $E(S, T)$ for the set of edges
with one endpoint in $S$ and one endpoint in $T$. 
For $S \subseteq V$, we write $\delta(S) = E(S, V \setminus S)$ for the \emph{cut induced by
$S$}. We will also refer to the partition $(S, V \setminus S)$ as a \emph{cut}; which
meaning we intend will be clear from the context.
If $| \, |S| - |V \setminus S| \, | \leq 1$, we will call
the cut $(S, V \setminus S)$ a \emph{bisection}.
Given $v \in V$ we will also write $\delta(v) = \delta(\{v\})$,
which is the set of \emph{edges incident to $v$}.

Let $F \subseteq E$. Given a
function $f : E \to \mathds{R}$,
we denote by $f(F) = \sum_{e \in F} f(e)$ the \emph{total value} of
$f$ on the set of edges $F$.
If $F = E(X, Y)$ for some sets $X, Y \subseteq V$, we will abuse notation
to write $f(X, Y) = f(F)$.

\subsection{The Class PLS}\label{section PLS}

For convenience, we summarize the formal definitions of
local search problems and the associated complexity class \PLS{}, as
devised by Johnson et al.\ \cite{johnsonHowEasyLocal1988}.

A \emph{local search problem} $P$ is defined by a set of instances $I$,
a set of feasible solutions~$F_I(x)$ for each instance $x \in I$,
a \emph{cost function} $c$ that maps pairs of a solution of $F_I(x)$ and an instance $x$ to $\mathds{Z}$,
and a \emph{neighborhood function}
$\mathcal{N}$ that maps a solution of $F_I(x)$ and an instance $x$ to a subset of $F_I(x)$.
Typically, the neighborhood function 
is constructed so that it is easy to compute some $s' \in \mathcal{N}(s, x)$ for
any given $s \in F_I(x)$.

This characterization of local search problems gives rise to the
\emph{transition graph} defined by an instance of such a problem.

\begin{definition}
    Given an instance $x \in I$ of a local search problem $P$, we define the
    \emph{transition graph} $T(x)$ as the directed graph with vertex set $F_I(x)$,
    with an edge from $s$ to~$s'$ if and only if $s' \in \mathcal{N}(s, x)$ and $c(s, x) < c(s', x)$ (assuming $P$ is
    a maximization problem; otherwise, we reverse the inequality).
    The height of a vertex $s$ in $T(x)$ is the length  of the shortest
    path from $s$ to a sink of $T(x)$.
\end{definition}

The class \PLS{} is defined to capture the properties of local search
problems that typically arise in practical applications. Formally,
$P$ is contained in \PLS{} if the following are all true:
\begin{enumerate}
    \item There exists a deterministic polynomial-time algorithm $A$ that, given an instance $x \in I$, computes
        some solution $s \in F_I(x)$.
        
    \item There exists a deterministic polynomial-time algorithm $B$ that, given $x \in I$
        and $s \in F_I(x)$, computes the value of $c(s, x)$.

    \item There exists a deterministic polynomial-time algorithm $C$ that, given $x \in I$
        and $s \in F_I(x)$, either computes a solution $s' \in \mathcal{N}(s, x)$ with
        $c(s', x) > c(s, x)$ (in the case of a maximization problem),
        or outputs that such a solution does not exist.
\end{enumerate}

Intuitively, algorithm~$A$ gives us some initial solution from which to start an optimization
process, algorithm~$B$ ensures that we can evaluate the quality of solutions efficiently, and
algorithm~$C$ drives the local optimization process by either determining that a solution
is locally optimal or otherwise giving us an improving neighbor. 
Based on the algorithms~$A$ and~$C$, one can define the ``standard algorithm problem'' for~$P$ as follows.

\begin{definition}
Let~$P$ be a local search problem and let~$I$ be an instance of~$\textsc{P}$.
Moreover, let~$s^*(I)$ be the unique local optimum obtained by starting with the solution outputted by algorithm~$A$ and replacing the current solution by the better solution outputted by~$C$, until reaching a local optimum.
The~\emph{standard algorithm problem} for~$P$ asks for a given instance~$I$ of~$P$ and a locally optimal solution~$s'$ for~$I$ with respect to~$\mathcal{N}$, whether~$s'$ is exactly the solution~$s^*(I)$.
\end{definition}

It was shown that for many local search problems the standard algorithm problem is~\PSPACE-complete~\cite{brauerComplexitySingleSwapHeuristics2017,michielsTheoreticalAspectsLocal2007,monienPowerNodesDegree2010,schafferSimpleLocalSearch1991}.


Given problems $P, Q \in \text{\PLS{}}$, we say that $P$ is \PLS{}-\emph{reducible}
to $Q$ (written $P \leq_\text{\PLS} Q$) if the following is true. 
\begin{enumerate}
    \item There exist polynomial-time computable functions $f,g$,
        such that $f$ maps instances $x$~of~$P$ to instances $f(x)$ of $Q$, and
        $g$ maps pairs $(\text{solution $s$ of $f(x)$}, x)$ to feasible solutions of $P$.
    \item If a solution $s$ of $f(x)$ is locally optimal for $f(x)$, then
        $g(s, x)$ is locally optimal for $x$.
\end{enumerate}

The idea is that, if $Q \in \text{\PLS{}}$ is efficiently solvable, then $P$ is also efficiently solvable:
simply convert an instance of $P$ to $Q$ using $f$, solve $Q$, and convert the resulting
solution back to a solution of $P$ using $g$. As usual in complexity theory,
if $P$ is \emph{complete} for \PLS{} and~$P \leq_\text{\PLS} Q$, then $Q$ is also
complete for \PLS{}.

In addition to this standard notion of a \PLS{}-reduction, Sch\"affer and Yannakakis
\cite{schafferSimpleLocalSearch1991}
defined so-called \emph{tight} reductions. Given \PLS{} problems $P$ and $Q$ and
a \PLS-reduction $(f, g)$ from $P$ to $Q$, the reduction is called \emph{tight} if for any
instance $x$ of $P$ we can choose a subset~$\R$ of the feasible solutions of
$f(x)$ of $Q$ such that:
\begin{enumerate}
    \item $\R$ contains all local optima of $f(x)$.
    \item For every feasible solution $s$ of $x$, we can construct a feasible
        solution $q \in \R$ of $f(x)$ such that $g(q, x) = s$.\
    \item Suppose the transition graph $T(f(x))$ of $f(x)$ contains
        a directed path from $s$ to~$s'$ such that $s, s' \in \R$, but all internal
        vertices lie outside of $\R$, and let $q = g(s, x)$ and~$q' = g(s', x)$.
        Then either $q = q'$, or the transition graph $T(x)$ of $x$ contains an
        edge from $q$ to $q'$.
\end{enumerate}
The set $\R$ is typically called the set of \emph{reasonable solutions to $f(x)$}. Here,
the intuition is that tight reductions make sure that the height of a vertex
$s$ of $T(f(x))$ is not
smaller than that of $g(s, x)$ in $T(x)$. Note that a reduction $(f, g)$ is trivially
tight if $T(f(x))$ is isomorphic to $T(x)$.

Tight \PLS{}-reductions have two desired properties \cite[Chapter 2]{aartsLocalSearchCombinatorial2003}.
Suppose $P$ reduces to~$Q$ via a tight reduction. First, if the standard algorithm problem
for $P$ is \PSPACE{}-complete, then the standard algorithm problem for $Q$ is \PSPACE{}-complete
as well. Second, if there exists an instance $x$ of $P$ such that there exists a solution
of $x$ that is exponentially far away from any local optimum, then such an instance exists 
for $Q$ as well. Note that this first property holds irrespective of the choices made by
the algorithm $C$ for $Q$ \cite{schafferSimpleLocalSearch1991}.

\subsection{Definitions of Local Search Problems}\label{sec:definitions}

We will be concerned with various local search problems. 
In the following we provide
a summary of the problems that appear in this paper, and
provide definitions for each. 
Some of the problems considered in this paper are not the most natural ones, but we need them as intermediate problems for our reductions.
Moreover, these problems might be useful to show \PLS{}-hardness of other problems having cardinality constraints.

Before introducing the problems themselves, we first provide a more abstract
view of the problems, since they have many important aspects in common. Each problem in the list
below is a type of partitioning problem, where we are given a finite set $S$ and are asked
to find the ``best'' partition of $S$ into $k$ sets (indeed, for all but one problem, we have
$k = 2$). What determines whether some partition is better than another varies;
this is determined by the cost function of the problem in question.

\begin{definition}
    Given a partition $\mathcal{P} = \{S_1, \ldots, S_k\}$ of $S$,
    a partition $\mathcal{P}'$ is a neighbor of~$\mathcal{P}$ in the \emph{\textsc{Flip} neighborhood}
    if $\mathcal{P}'$ can be obtained by
    moving exactly one element from some~$S_i \in \mathcal{P}$ to some other~$S_j \in \mathcal{P}$.
    In other words, if $\mathcal{P}' = \{S_1, \ldots, S_i', \ldots, S_j', \ldots, S_k\}$
    where for some $v \in S_i$ we have~$S_i' = S_i \setminus \{v\}$ and $S_j' = S_j \cup \{v\}$.
\end{definition}

The \textsc{Flip} neighborhood as defined above is perhaps the simplest neighborhood
structure for a partitioning problem.
For each problem in the list below, we consider only
the \textsc{Flip} neighborhood in this paper. 
Recall that the \textsc{Flip} heuristic for \textsc{$k$-Means} is also referred to as the \textsc{Hartigan--Wong} method~\cite{hartiganAlgorithm136KMeans1979}.

\problemdef{Max Cut}
{A graph~$G = (V, E)$ with non-negative edge weights~$w: E \to \mathds{Z}_{\geq 0}$.}
{A partition~$(X,Y)$ of~$V$ such that~$w(X, Y)$ is maximal.}

We will mainly be concerned with several variants of \textsc{Max Cut}.
For some fixed integer~$d$, by \textsc{Max Cut-$d$} we denote the restriction of the problem to graphs with maximum degree~$d$.
In \textsc{Densest Cut}, one aims to maximize $\frac{w(X, Y)}{|X|\cdot|Y|}$ rather than just $w(X, Y)$. The minimization
version of this problem is called \textsc{Sparsest Cut}.
The problem \textsc{Odd Max Bisection} is identical to \textsc{Max Cut}, with the added restrictions that the number of vertices must
be odd and that the two halves of the cut differ in size by exactly one. The minimization version
of the problem is called \textsc{Odd Min Bisection}.

The definitions of \textsc{Odd Max/Min Bisection} are somewhat unconventional, as one usually
considers these problem with an even number of vertices and
with the \textsc{Swap} neighborhood, where
two solutions are neighbors if one can be obtained from the other by
swapping a pair of vertices between parts of the partition. 
Hardness of \textsc{Max Bisection/Swap} was shown by Sch\"affer and Yannakakis \cite{schafferSimpleLocalSearch1991}
in a short reduction from \maxcut{}, providing as a corollary also a simple hardness proof
for the Kernighan-Lin heuristic \cite{kernighanEfficientHeuristicProcedure1970}.
The reason we require the \textsc{Flip} neighborhood is that we aim to reduce
this problem to \emaxcut{}, where we run
into trouble when we use the \textsc{Swap} neighborhood (see \Cref{sec:reduce to bisection} for details).

\iflong
\problemdef{Squared Euclidean Max Cut}
{A complete graph~$G=(V,E)$ with non-negative edge weights~$w: E \to \mathds{Z}_{\geq 0}$, such that there exists an
injection~$f: V \to \mathds{R}^d$ such that $w(uv) = \|f(u) - f(v)\|^2$ for all edges $uv \in E$.}
    {A partition~$(X,Y)$ of~$V$ such that $w(X, Y)$ is maximal.}
\fi
    
\textsc{Euclidean Max Cut} is defined similarly; the only difference is that
\iflong
we have~$w(uv) = \|f(u) - f(v)\|$.
\textsc{Squared Euclidean Max Cut} is more commonly defined in terms of an input point set~$\X \subseteq \mathds{R}^d$, with
solutions being partitions~$(X, Y)$ of~$\X$ with cost~$\sum_{x \in X}\sum_{y \in Y} \|x - y\|^2$.
However, in this formulation, membership of \emaxcut{} in \PLS{} is no longer necessarily
true\footnote{
    The problem \textsc{Sum of Square Roots} asks whether for given numbers $a_i$ and $k_i$ the
    linear combination~$\sum_i k_i \sqrt{a_i}$ evaluates to zero (among other variants).
    It is not currently known whether this problem can be solved in time polynomial in the number of
    bits needed to represent the given linear combination, although a randomized algorithm can solve
    the problem in polynomial time \cite{blomerComputingSumsRadicals1991}. Since computing the cost of a solution
    and deciding local optimality for \textsc{(Squared) Euclidean Max Cut/Flip} and \khartigan{} requires
    the computation of such linear combinations when formulated in the traditional way, it is not clear
    whether these formulations are in \PLS{}. Nevertheless, our reductions still imply \PLS{}-hardness
    for these traditional formulations, as well as the existence of certain hard instances.
}. 
Hence, we opt for the above graph-theoretical formulation.
\fi

\iflong
\problemdef{$k$-Means}
{A complete graph~$G=(V,E)$ with non-negative edge weights~$w: E \to \mathds{Z}_{\geq 0}$, such that there exists a 
injection~$f: V \to \mathds{R}^d$ such that $w(uv) = \|f(u) - f(v)\|^2$ for all edges $uv \in E$.}
{A partition~$(C_1,\ldots, C_k)$ of~$V$ such that~$\sum_{i=1}^k \frac{1}{|C_i|}\sum_{u, v \in C_i} w(uv)$ is minimal.}
\fi

The sets~$\{C_1, \ldots, C_k\}$ are called \emph{clusters}.
\iflong
As with \textsc{Squared Euclidean Max Cut}, 
in the literature \textsc{$k$-Means} is usually specified with an input point
set $\X \subseteq \mathds{R}^d$, and the cost of a partition $\{C_1,\ldots, C_k\}$ of $\X$ is given as
$\sum_{i=1}^k \sum_{x \in C_i} \|x - \cm(C_i)\|^2$, where~$\cm(C_i)$ denotes to \emph{centroid} of~$C_i$. 
This is equivalent to our formulation when
the coordinates of the points are algebraic numbers
\cite[Proposition 1]{pyatkinNPHardnessBalancedMinimum2017}.
In \Cref{lemma:densestcut_to_hartigan,lemma:k-means} we will also use this notation.
\fi

\problemdef{Positive Not-All-Equal $k$-Satisfiability (Pos NAE $k$-SAT)}
{A boolean formula with clauses of the form $\nae(x_1, \ldots, x_\ell)$ with $\ell \leq k$, where each clause is satisfied
if its constituents, all of which are positive,
are neither all true nor all false. Each clause $C$ has a weight~$w(C)\in\mathds{Z}$.}
{A truth assignment of the variables such that the sum of the weights of the satisfied clauses is maximized.}

In \textsc{Odd Half Pos NAE $k$-SAT}, additionally the number of variables is odd and it is required that the number of \true{} variables and the number of \false{} variables in any solution differ by exactly one. This is analogous to the relationship
between \textsc{Max Cut} and \textsc{Odd Max Bisection}.

\subsection{Strategy}\label{sec:strategy}

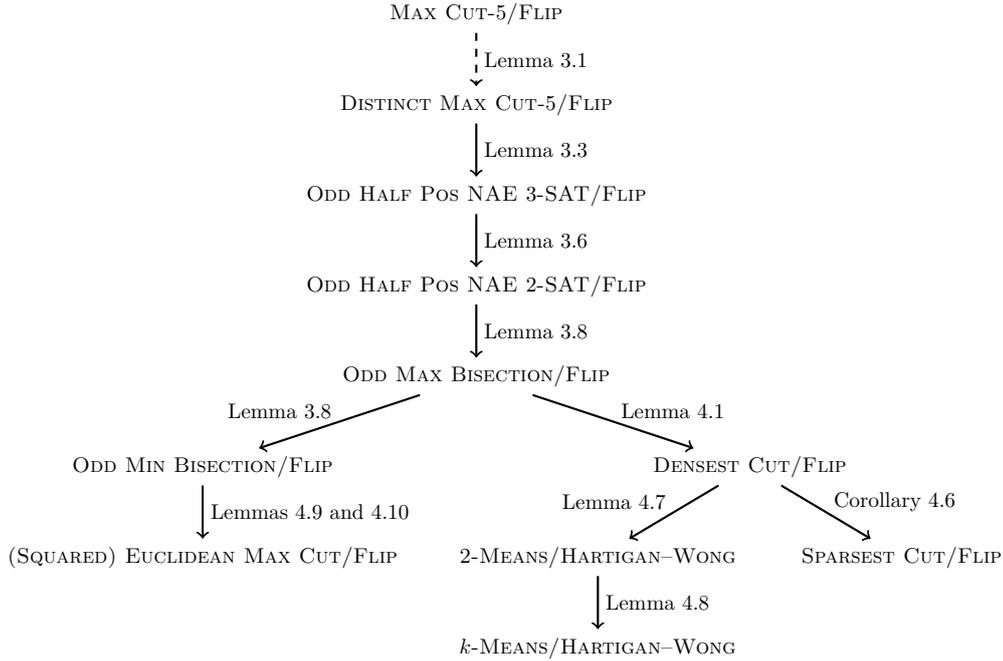
\begin{figure}[t]
\begin{center}
\scalebox{0.8}{
\begin{tikzpicture}[thick]
    \node (maxcut5) at (0, 0) {\maxcutd[5]};
    \node (maxcut5_distinct) at (0, -1.5) {\maxcutdistinctd[5]};
    \node (halfposnae3sat) at (0, -3) {\halfposnaesat[3]};
    \node (halfposnae2sat) at (0, -4.5) {\halfposnaesat[2]};
    \node (maxbisection) at (0, -6) {\maxbisection};
    \node (minbisection) at (-4.5, -7.5) {\minbisection};
    \node (squaredeucmaxcut) at (-4.5, -9) {\textsc{(Squared) Euclidean Max Cut/Flip}};
    
    \node (densestcut) at (4.5, -7.5) {\densestcut};
    \node (hartigan) at (2, -9) {\hartigan};
    \node (khartigan) at (2, -10.5) {\khartigan};
    
    \node (sparsestcut) at (7, -9) {\sparsestcut};
    
    \draw[dashed, ->] (maxcut5) -- (maxcut5_distinct) node[midway, right] {\Cref{lemma:maxcut_distinct_costs}};
    \draw[->] (maxcut5_distinct) -- (halfposnae3sat) node[midway, right] {\Cref{lemma:halfposnaesat}};
    \draw[->] (halfposnae3sat) -- (halfposnae2sat) node[midway, right] {\Cref{lemma:reduce3to2}};
    \draw[->] (halfposnae2sat) -- (maxbisection) node[midway, right] {\Cref{lemma:naetobisection}};
    \draw[->] (maxbisection) -- (minbisection) node[midway, left, yshift=5] {\Cref{lemma:naetobisection}};
    \draw[->] (minbisection) -- (squaredeucmaxcut) node[midway, right] {\Cref{lemma:bisection_to_emaxcut,lemma:bisection_to_emaxcut_nosquare}};
    \draw[->] (maxbisection) -- (densestcut) node[midway, right, yshift=5] {\Cref{lemma:bisection_to_densestcut}};
    \draw[->] (hartigan) -- (khartigan) node[midway, right] {\Cref{lemma:k-means}};
    \draw[->] (densestcut) -- (hartigan) node[midway, left, yshift=5] {\Cref{lemma:densestcut_to_hartigan}};
    
    \draw[->] (densestcut) -- (sparsestcut) node[midway, right, yshift=5] {\Cref{cor:sparsest-cut}};

\end{tikzpicture}
}
\end{center}
\caption{Graph of the \PLS-reductions used in this paper. Reductions represented by solid lines are
tight, reductions represented by dashed lines are not.}
\label{fig:reduction_graph}
\end{figure}

Both \textsc{Squared Euclidean Max Cut} and \textsc{$k$-Means} are \NP-hard~\cite{ageevComplexityWeightedMaxcut2014,ADHP09}.
The reductions used to prove this are quite similar, and can be straightforwardly adapted into
\PLS{}-reductions: In the case of \emaxcut{}, we obtain a reduction
from \minbisection{}, while for \khartigan{} we obtain a reduction from
\densestcut{}. The latter reduction even works for $k = 2$. These results
are given in \Cref{lemma:densestcut_to_hartigan} (\textsc{$k$-Means}), and \Cref{lemma:bisection_to_emaxcut_nosquare,lemma:bisection_to_emaxcut} (\textsc{(Squared) Euclidean Max Cut}).

What remains then is to show that the problems we reduce from are also
\PLS{}-complete, which takes up the bulk of the work. \Cref{fig:reduction_graph}
shows the reduction paths we use. 

The starting point will be the \PLS{}-completeness of
\maxcutd[5], which was shown by Els\"asser and Tscheuschner \cite{elsasserSettlingComplexityLocal2011}.
An obvious next step might then be to reduce from this problem to
\textsc{Max Bisection/Swap}, and then further reduce to \maxbisection{}.
Unfortunately, this turns out to be rather difficult, as the extra power afforded
by the \textsc{Swap} neighborhood is not so easily reduced back to the \textsc{Flip}
neighborhood. Using this strategy, we can only obtain \PLS{}-hardness of the
\textsc{2-Flip} neighborhood for \textsc{Squared Euclidean Max Cut},
where two points may flip in a single iteration.

We thus take a detour through \halfposnaesat[3] in \Cref{lemma:halfposnaesat}, which then reduces down to
\halfposnaesat[2] in \Cref{lemma:reduce3to2} and finally to \maxbisection{}
in \Cref{lemma:naetobisection},
using a reduction by Sch\"affer and Yannakakis \cite{schafferSimpleLocalSearch1991}.

From this point, hardness of \emaxcut{} (and with a little extra work,
\emaxcutnosq{}) is easily obtained. For \khartigan, we need some more effort,
as we still need to show hardness of \densestcut{}. Luckily, this can be done by reducing from \maxbisection{} as well,
as proved in \Cref{lemma:bisection_to_densestcut}.

\section{Reduction to Odd Min/Max Bisection}
\label{sec:reduce to bisection}

The goal of this section is to obtain \PLS{}-completeness of \textsc{Odd Min/Max Bisection/Flip},
from which we can reduce further to our target problems; see \Cref{fig:reduction_graph}.
We will first construct a \PLS{}-reduction from \maxcutd[5] to \halfposnaesat[3]
in \Cref{lemma:halfposnaesat}. 

A subtlety is that the reduction
only works when we assume that the \textsc{Max Cut-5} instance we reduce from
has distinct costs for any two neighboring solutions. The following lemma ensures that we can make
this assumption.

\begin{lemma}\label{lemma:maxcut_distinct_costs}
    \maxcutdistinctd[5] is \PLS{}-complete. 
    More precisely, there exists a \PLS{}-reduction from \maxcutd[5] to \maxcutdistinctd[5].
\end{lemma}

\iflong

\begin{proof}
    Let $I$ be an instance of \maxcutd[5] on a graph $G = (V(G), E(G))$. Let $S$ be the set of feasible solutions
    to $I$, and write $c(s, I)$ for the cost of a solution $s$ to $I$.
    By definition of the local search problem, the edge weights are integral, and so the cost takes only
    integer values. Thus, the smallest nonzero difference $\Delta_{\min{}}$ between any two neighboring solutions is at least
    $1$. 
    
    First, we multiply the value of all edge weights by $10$; since the cost function
    is linear in the edge weights
    this yields an equivalent instance, and $\Delta_{\min{}}$ becomes at least $10$. Next, for every
    vertex $v \in V$ with even degree we add a dummy vertex to $G$ and connect it to~$v$ with
    an edge of zero weight to create a graph $H = (V(H), E(H))$.
    Every vertex of $H$ has odd degree, and~$H$ has at most
    $2\cdot |V|$ vertices.
    Since the dummy vertices only have zero-weight incident edges, this instance is still
    equivalent to the original instance; for any cut $(A, B)$ in~$H$, just take $\left(A \cap V(G), B \cap V(G)\right)$
    as a cut for $G$. This cut then has the same cost as~$(A, B)$.
    Observe also that $H$ still has maximum degree five, since only vertices with degree at most
    four get an extra incident edge.
    
    We modify $I$ into an instance $I'$ by adding $1$ to all edge weights.
    
    \proofsubparagraph{Correctness.}
    Let $s$ and $s'$ be two neighboring solutions to $I'$.
    The neighbor $s'$ is obtained by flipping some vertex.
    Observe that the only edge weights that enter $c(s', I) - c(s, I)$ are those incident
    to this vertex. Thus, since every vertex of $H$ has odd degree, we have an odd
    number of terms. Since every term is changed by exactly $1$ in $I'$, if
    $c(s', I) = c(s, I)$, then~$c(s', I') \neq c(s, I')$. Moreover, if $c(s', I) \neq c(s, I)$,
    then the two differ by at least 10, while the value of $c(s', I') - c(s, I')$ differs
    from $c(s', I) - c(s, I)$ by at most 5. Hence, $c(s, I) \neq c(s', I)$ implies
    $c(s, I') \neq c(s', I')$, and so
    any two neighboring solutions to $I'$ have distinct cost.
    
    We claim that a locally optimal solution to $I'$ is also locally
    optimal for $I$. We prove the contrapositive.
    Suppose a solution $s$ to $I$ is not locally optimal, so there exists a neighbor~$s'$ of~$s$
    such that $c(s', I) - c(s, I) \geq 10$.
    Following the same argument as in the previous paragraph, $c(s', I') - c(s, I')$~differs
    from~$c(s', I) - c(s, I)$ by at most $5$. Hence,~$c(s', I') - c(s, I') \geq 5$.
    Therefore, $s$ is also not locally optimal for $I'$. 
\end{proof}

\fi

\iflong
Note that the reduction used in \Cref{lemma:maxcut_distinct_costs} is not tight.
\fi
Hence, to prove the last two items of \Cref{thm:hardinstances}, simply applying
the reductions from \Cref{fig:reduction_graph} is not sufficient, as these
properties (viz.\ \PSPACE{}-completeness of the standard algorithm problem and the existence
of certain hard instances)
do not necessarily hold for \maxcutdistinctd[5]. We must instead
separately prove that they hold for this problem.
To accomplish this, we recall a construction
by Monien and Tscheuschner \cite{monienPowerNodesDegree2010}
that shows these properties for \maxcutd[4]{}. It can be verified
straightforwardly that the construction they use is already an instance of
\maxcutdistinctd[4]\iflong{}, as we show\fi{}. 
\iflong
\fi

\iflong

To verify this claim, we first state some terminology used by Monien and Tscheuschner.
Let~$G = (V, E)$ be a graph of maximum degree four. Let~$u \in V$,
and let~$\{a_u, b_u, c_u, d_u\}$ be the edges incident to~$u$; if~$u$ has degree
less than four, we consider the weights of the non-existent edges to be zero.
Assume in the following
that~$w(a_u) \geq w(b_u) \geq w(c_u) \geq w(c_u)$.
We distinguish three types of vertices. A vertex~$u$ is of:

\begin{enumerate}[Type I,]
 \item if~$w(a_u) > w(b_u) + w(c_u) + w(d_u)$;
 \item if~$w(a_u) + w(d_u) > w(b_u) + w(c_u)$ and~$w(a_u) < w(b_u) + w(c_u) + w(d_u)$;
 \item if~$w(a_u) + w(d_u) < w(b_u) + w(c_u)$.
\end{enumerate}

While general weighted graphs may contain vertices not of these three types,
the graphs constructed by Monien and Tscheuschner only contain Type I, II and
III vertices.

Given a cut of~$V$, a vertex~$u$ is called \emph{happy} if flipping~$u$
cannot increase the weight of the cut. Monien and Tscheuschner note that
a vertex of type I is happy if and only if the edge~$a_u$ is in the cut, a vertex of type II
is happy if and only if~$a_u$ and at least one of~$\{b_u, c_u, d_u\}$ is in the cut,
and a vertex of type III is happy if and only if at least two of $\{a_u, b_u, c_u\}$ are in
the cut.

\begin{lemma}
    Let the graph above define an instance~$x$ of \maxcutd[4]{}, and let~$T(x)$ be its transition graph.
    Suppose~$s$ and~$s'$ are neighbors in~$T(x)$ which differ by the flip of
    a vertex of type I, type II, or type III. Then~$c(s, x) \neq c(s', x)$.
\end{lemma}

\begin{proof}
    Let~$u$ be the vertex in which~$s$ and~$s'$ differ, with incident edges
    $\{a_u, b_u, c_u, d_u\}$. Assume the weights of these edges are ordered as
    $w(a_u) \geq w(b_u) \geq w(c_u) \geq w(d_u)$.
    Furthermore, we assume w.l.o.g.\ that~$u$
    is happy in~$s$. We consider the three possibilities for the type of $u$.
    
    \proofsubparagraph{Type I.} If $u$ is happy in $s$, then
    $a_u$ is in the cut. Thus, we have~$c(s, x) - c(s', x) \geq w(a_u) - (w(b_u) + w(c_u) + w(d_u)) > 0$. The claim
    follows.
    
    \proofsubparagraph{Type II.} Since~$u$ is happy in~$s$, we know that~$a_u$ and
    at least one of~$\{b_u, c_u, d_u\}$ are in the cut in~$s$.
    Assume that~$d_u$ is in the cut with~$a_u$.
    Then~$c(s, x) - c(s', x) \geq w(a_u) + w(d_u) - (w(b_u) + w(c_u)) > 0$. Now observe that
    if~$d_u$ is swapped out for~$b_u$ or~$c_u$, then~$c(s, x) - c(s', x)$ cannot decrease,
    since~$w(d_u) \leq w(c_u) \leq w(b_u)$ by assumption. Thus, flipping
    $u$ changes the weight of the cut if $u$ is of type II.
    
    \proofsubparagraph{Type III.} Since~$u$ is happy in $s$, at least two of
    $\{a_u, b_u, c_u\}$ are in the cut. We distinguish two cases.
    \begin{description}
        \item[Case 1.] Assume first that~$a_u$ is not in the cut, so that
            $b_u$ and $c_u$ are both in the cut.
             Then~$c(s, x) - c(s', x) \geq w(b_u) + w(c_u) - (w(a_u) + w(d_u)) > 0$.
        \item[Case 2.] Now we assume that~$a_u$ is in the cut.
            Then at least one of $b_u$ and $c_u$ is in the cut with $a_u$.
            Thus, since $w(a_u) \geq w(b_u) \geq w(c_u)$, we have
            $c(s, x) - c(s', x) \geq w(a_u) + w(c_u) - (w(b_u) + w(d_u))
            \geq w(b_u) + w(c_u) - (w(a_u) + w(d_u)) > 0$. Again, flipping
            $u$ strictly changes the weight of the cut. 
    \end{description}

    This covers all possibilities, and so we conclude that flipping a vertex
    of any of the three types defined above strictly changes the weight of the cut.
\end{proof}

Since the graphs constructed by Monien and Tscheuschner only contain vertices of these three
types, we can conclude that these graphs are instances of \textsc{Distinct Max Cut-4} and
thus also of \textsc{Distinct Max Cut-5}.

In the remainder of \Cref{sec:reduce to bisection} and subsequently in \Cref{sec:reduce to clustering} we will obtain
the \PLS{}-reductions shown in \Cref{fig:reduction_graph}. Since the reductions from
\maxcutdistinctd[5] are tight, and by the properties of tight \PLS{}-reductions and the
results of Monien and Tscheuschner,
\Cref{thm:hardinstances} then follows.
\fi

In the remainder of this work, we present a sequence of tight reductions starting from \maxcutdistinctd[5]{} to all of our other
considered problems.
First, we reduce from \maxcutdistinctd[5] to \halfposnaesat[3].

\begin{lemma}\label{lemma:halfposnaesat}
    There exists a tight \PLS{}-reduction from \maxcutdistinctd[5]{} to \halfposnaesat[3].
\end{lemma}

As mentioned in \Cref{sec:strategy}, it may seem more straightforward to reduce from
\textsc{Min Bisection/Swap} to \minbisection{}.
The problem with this approach
is that a solution to \minbisection{} can be locally optimal for two reasons: either 
no vertex can flip to obtain a cut of larger weight, or the vertices that could are in the smaller
part of the partition. This makes the \textsc{Flip} neighborhood much less powerful than \textsc{Swap} in this problem variant;
we were thus not able to find a direct reduction from \textsc{Min Bisection/Swap}. Instead, 
we apply a new technique that allows us to prove \PLS{}-hardness for this very restricted problem.

We first prove \PLS{}-hardness of 
\halfposnaesat[3], and subsequently use existing reductions to obtain hardness of \minbisection{}.
With the expressiveness of this \textsc{SAT} variant we gain a great deal of freedom to handle the problem restrictions.
The main challenge is in
encoding the restriction that the number of true and false variables must differ by exactly one
without weakening the neighborhood.

\iflong{}
We achieve this in two steps.
First, by adding a large number of extra variables, we can embed any solution to our original instance of
\textsc{Max Cut-5} into \textsc{Odd Half Pos NAE 3-SAT} while still satisfying the size constraints.
Second, we must restore the power of the \textsc{Flip} neighborhood in this instance. 
By adding new clauses and variables to the instance, we simulate a flip in the original
\textsc{Max Cut-5} instance by a sequence of at most two flips in its image instance. 
This requires carefully constructing clauses
which encode subsets of neighborhoods of vertices in the original \textsc{Max Cut-5} instance. Since each vertex in this
instance has bounded degree, there are a constant number of subsets of the neighborhood of any vertex, yielding
a polynomial-time reduction.

As far as we are aware, this technique for overcoming size constraints in local search problems is novel.
We believe that it may be useful to prove \PLS{}-hardness results for simple heuristics for other size-constrained problems,
such as balanced clustering problems.
\fi

Next, we briefly sketch and motivate some of the ideas in the reduction in more detail\iflong{}, before proceeding to the proof\fi{}.
See also \Cref{fig:halfposnaesat}.

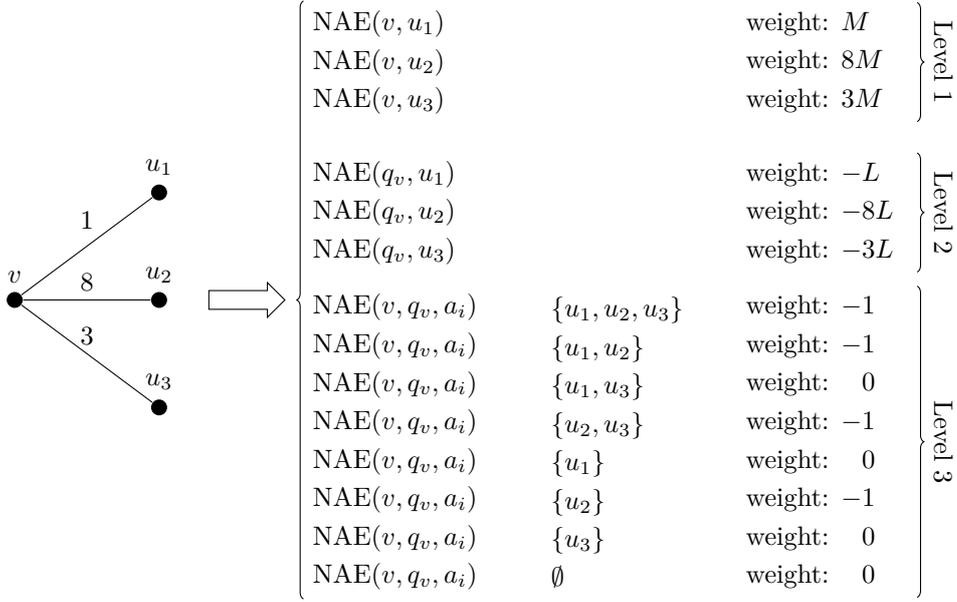
\begin{figure}[t]
\centering
\begin{tikzpicture}[scale=0.95]

    \node[anchor=west, align=left] (top) at (4, 2.1)
    {
        $\nae(v, u_1)$\\[1.5pt]
        $\nae(v, u_2)$\\[1.5pt]
        $\nae(v, u_3)$
    };
    \node[anchor=west, align=left](level1) at (10, 2.1) 
    {
        weight: $M$\\[1.5pt] 
        weight: $8M$\\[1.5pt] 
        weight: $3M$
    };
    
    \coordinate (p) at (level1.north east);
    \coordinate (q) at (level1.south east);
    \draw[
            decoration={
        brace,
    }, decorate] let \p1=(p), \p2=(q) in (12.5, \y1) -- (12.5, \y2) node[sloped, pos=0.5, above=0.1cm] {Level 1};
    
    \node[anchor=west, align=left] at (4, 0)
    {
        $\nae(q_v, u_1)$\\[1.5pt]
        $\nae(q_v, u_2)$\\[1.5pt]
        $\nae(q_v, u_3)$
    };
    \node[anchor=west, align=left](level2) at (10, 0) 
    {
        weight: $-L$\\[1.5pt] 
        weight: $-8L$\\[1.5pt] 
        weight: $-3L$
    };
    
    \coordinate (p) at (level2.north east);
    \draw[
            decoration={
        brace,
    }, decorate] let \p1=(p) in (12.5, \y1) -- (12.5, {-\y1}) node[sloped, pos=0.5, above=0.1cm] {Level 2};

    \node[anchor=west, align=left] (bot) at (4, -3.2) 
    {
        $\nae(v, q_v, a_i)$\\[1.5pt] 
        $\nae(v, q_v, a_i)$\\[1.5pt]
        $\nae(v, q_v, a_i)$\\[1.5pt] 
        $\nae(v, q_v, a_i)$\\[1.5pt] 
        $\nae(v, q_v, a_i)$\\[1.5pt] 
        $\nae(v, q_v, a_i)$\\[1.5pt] 
        $\nae(v, q_v, a_i)$\\[1.5pt] 
        $\nae(v, q_v, a_i)$
    };
    \node[anchor=west, align=left] at (7.3, -3.2) 
    {
        ${\{u_1, u_2, u_3\}}$\\[1.5pt] 
        ${\{u_1, u_2\}}$\\[1.5pt]
        $\{u_1, u_3\}$\\[1.5pt] 
        ${\{u_2, u_3\}}$\\[1.5pt] 
        ${\{u_1\}}$\\[1.5pt] 
        ${\{u_2\}}$\\[1.5pt] 
        ${\{u_3\}}$\\[1.5pt] 
        $\emptyset$
    };
    \node[anchor=west, align=left] (level3) at (10, -3.2) 
    {
        weight: $-1$\\[1.5pt] 
        weight: $-1$\\[1.5pt] 
        weight: $\phantom{-}0$\\[1.5pt] 
        weight: $-1$\\[1.5pt] 
        weight: $\phantom{-}0$\\[1.5pt] 
        weight: $-1$\\[1.5pt] 
        weight: $\phantom{-}0$\\[1.5pt] 
        weight: $\phantom{-}0$
    };
    
    \coordinate (p) at (level3.north east);
    \coordinate (q) at (level3.south east);
    \draw[
            decoration={
        brace,
    }, decorate] let \p1=(p),\p2=(q) in (12.5, \y1) -- (12.5, \y2) node[sloped, pos=0.5, above=0.1cm] {Level 3};

    \draw[
            decoration={
        brace,
    }, decorate] (bot.south west) to coordinate[pos=.5] (o) (top.north west) node[sloped, pos=0.5, above=0.1cm] {};

    \gettikzxy{(o)}{\ox}{\oy}
    
    \vertex[label=$v$] (v) at (0, \oy) {};
    \vertex[label=$u_1$](u1) at (2, {\oy+1.5cm}) {};
    \vertex[label=$u_2$](u2) at (2, {\oy+0}) {};
    \vertex[label=$u_3$](u3) at (2, {\oy-1.5cm}) {};

    \draw (v) -- node[above=3pt] {$1$} ++ (u1);
    \draw (v) -- node[above] {$8$} ++ (u2);
    \draw (v) -- node [above] {$3$} ++ (u3);

    \node[single arrow, draw=black,
      minimum width = 10pt, single arrow head extend=3pt,
      minimum height=10mm] at (3.15, \oy) {}; 
\end{tikzpicture}

\caption{Schematic overview of the reduction used in the proof of \Cref{lemma:halfposnaesat}.
On the left we have a vertex $v \in V$ and its neighbors $\{u_1, u_2, u_3\}$ in a \textsc{Max Cut} instance,
with weights on the edges between $v$ and its neighbors.
The $\nae$ clauses on the right are the clauses constructed from~$v$.
In the actual reduction, these clauses are added for all level 3 variables.
The right-most column shows the weights assigned to the clauses. 
The middle column shows for the level 3 clauses which subset of~$N(v)$ corresponds to
which clause.
The constants
$L$ and $M$ are chosen so that $1 \ll L \ll M$.}
\label{fig:halfposnaesat}
\end{figure}

\proofsubparagraph{Sketch of Proof for \Cref{lemma:halfposnaesat}.} We first embed the
\textsc{Distinct Max Cut-5} instance,
given by a weighted graph $G = (V, E)$,
in \textsc{Pos NAE SAT}. This can be done rather straightforwardly, by a reduction
used by Sch\"affer and Yannakakis \cite{schafferSimpleLocalSearch1991}: each vertex
becomes a variable, and an edge $uv$ becomes a clause $\nae(u, v)$. This instance
is directly equivalent to the original \textsc{Distinct Max Cut-5} instance. We call these
variables the \emph{level 1 variables}, and the clauses the \emph{level 1 clauses}.
A level 1 clause $\nae(u, v)$ gets a weight $M\cdot w(uv)$ for some large integer $M$.

A solution to the original \textsc{Distinct Max Cut-5} instance is obtained by
placing the \true{} level 1 variables in a feasible truth assignment on one side of the cut,
and the \false{} level 1 variables on the other side.

Given the reduction so far, suppose we have some locally optimal feasible truth assignment~$s$.
We partition the variables into the sets
$T$ and $F$ of \true{} and \false{} variables; thus, $(T, F)$ is the cut induced by $s$.
Suppose $|T| = |F| + 1$. If
no level 1 variable can flip from $T$ to $F$, then also no vertex can flip
from $T$ to $F$ in the induced cut.
However, we run into trouble when there exists some $v \in F$ that can flip in the cut.
Since $|F| < |T|$, we are not allowed to flip the level 1 variable $v$, and so
the truth assignment may be locally optimal even though the induced cut is not.

To deal with this situation, we will introduce two more levels of variables and clauses.
The weights of the clauses at level $i$ will be scaled so that they are
much larger than those at level $i+1$. In this way,
changes at level $i$ dominate changes at level $i+1$, so that
the \textsc{Distinct Max Cut-5} instance can exist independently at level 1.

For each vertex $v \in V$, we add a variable $q_v$ to the instance,
and for each $u \in N(v)$, we add a clause~$\nae(q_v, u)$ with weight proportional to
$-w(uv)$.
We call these variables the \emph{level 2 variables}, and
these clauses the \emph{level 2 clauses}.

Finally, we add $N = 2n+1$ more variables $\{a_i\}_{i=1}^N$,
which we call the \emph{level 3 variables}. The number $N$
is chosen so that for any truth assignment such that
the number of \true{} and \false{} variables differ by one, there must
exist a level 3 variable in the larger of the two sets. We then add more
clauses as follows: for each level 3 variable $a_i$, for each $v \in V$,
for each~$Q \subseteq N(v)$, we add a clause $C_i(v, Q) = \nae(v, q_v, a_i)$.
We give this clause a weight of~$-1$ if and only if $v$ can flip when each 
of the vertices in $Q$ are present in the same half of the cut as~$v$. We call these
the \emph{level 3 clauses}

Now consider the aforementioned situation, where a truth assignment
$s$ is locally optimal, but there exists some $v \in V \cap F$ that can flip
in the induced cut.
Carefully investigating the structure of such a locally optimal truth assignment shows that
some level 2 or level~3 variable can flip for a strict improvement
in the cost. This contradicts local optimality, and so we
must conclude that locally optimal truth assignments induce locally optimal
cuts, satisfying the essential property of \PLS{}-reductions. \hfill$\blacksquare$
\smallskip

\iflong

\begin{proof}[Proof of \Cref{lemma:halfposnaesat}]
    We now proceed to state the reduction formally and prove its correctness.
    Let $G = (V, E)$ be an instance of \textsc{Distinct Max Cut-5}
    with weights $w$ on the edges, and let~$n = |V|$. We write $\Delta_{\min{}}$ for the smallest absolute difference in cost
    between any two neighboring cuts,
    and~$\Delta_{\max{}}$ for the largest. By definition of \textsc{Distinct Max Cut-5},
    we have~$\Delta_{\min{}} \geq 1$.
    Note that the values~$\Delta_{\min{}}$ and~$\Delta_{\max{}}$ can be computed in polynomial time for a given instance
    of \maxcutdistinctd[5]. Finally, we define $N = 2n + 1$, $L = 2^6 N$, and
    $M = 5\cdot\Delta_{\max{}} \cdot L + 2^5N + 1$.

    To begin, we embed the \textsc{Max Cut} instance in our
    \textsc{Half Pos NAE 3-Sat} instance.
    For each vertex $v \in V$
    we create a variable, and for each edge $e = uv \in E$ we add a clause~$C_e = \nae(u, v)$ with weight $w(C_e) = M\cdot w(e)$.
    We call these variables the \emph{level 1 variables} and these clauses the \emph{level 1 clauses}.
    A truth assignment of
    the variables $V$ naturally induces a cut in $G$.
    Then $C_e$ is satisfied if and only if $u$ and $v$ have opposite truth values,
    which is the case if and only if $e$ is in the corresponding cut. The
    total weight of all satisfied clauses is exactly the weight of the cut times $M$,
    and a flip of a variable corresponds to a flip of a vertex from
    one side to the other.

    To this instance, we add for each vertex $v \in V$ a variable
    $q_v$; we call the variables $\{q_v\}_{v \in V}$ the \emph{level 2 variables}.
    We additionally add $N=2n+1$ more variables $\{a_1, \ldots, a_{N}\}$,
    which we call the \emph{level 3 variables}.
    Now, consider a vertex $v \in V$ and its neighborhood~$N(v)$.
    Observe that the change in cost due to a flip of $v$ in the given
    \textsc{Max Cut} instance is determined fully by
    which of its neighbors are in the same half of the cut as $v$.
    For each $Q \subseteq N(v)$ and each level 3 variable, we add a clause 
    \[
        C_i(v, Q) = \nae\left(v, q_v, a_i\right).
    \]
    See \Cref{fig:halfposnaesat} for
    an example. We call these clauses the \emph{level 3 clauses}.
    We set the weight of $C_i(v, Q)$ equal to
    \[
        w(C_i(v, Q)) = \begin{cases}
            -1 & \text{if } \sum_{u \in Q} w(vu) - \sum_{u \in N(v) \setminus Q} w(vu) > 0, \\
            0 & \text{otherwise.}
        \end{cases}
    \]
    Note the minus sign. By this construction, $C_i(v, Q)$ has nonzero weight if and only if
    the vertex $v \in V$ can flip in the \textsc{Max Cut} instance when its neighbors
    in $Q$ are in the same half of the cut as $v$.
    
    Observe that
    each subset of $N(v)$ yields a clause consisting of the same three variables.
    We consider these clauses to be distinct, rather than merging them into a single
    clause of the same total weight, since this simplifies the verification of the reduction later on.
    
    For each $v \in V$ we also add the clauses
    \[
        C(q_v, u) = \nae(q_v, u)
    \]
    for each $u \in N(v)$. Each clause
    $C(q_v, u)$ has weight equal to $-L\cdot w(uv)$ (note again the minus sign).
    We call these clauses the \emph{level 2 clauses}.
   
    We only add a polynomial number of variables and clauses.
    Let the set of all variables be $S$.
    The total number of variables is then $|S| = 2n + N = 4n+1$.
    Observe that this choice of $N$
    ensures that any feasible truth assignment of $S$ with $T$ the true
    variables and $F$ the false variables,
    we must have at least one level 3 variable in the larger of $T$ and $F$.
    
    \proofsubparagraph{Correctness.}
    It now remains to show that a locally optimal truth assignment induces
    a locally optimal cut.
    In the following, given a truth assignment~$s$, we will denote 
    by~$T(s)$ the set of variables assigned \true{} by $s$ and with~$F(s)$
    the set assigned \false{} by~$s$. 
    If the truth assignment~$s$ is clear from the context, we may simply write~$T$ and~$F$.
    
    We first require the following claim.

    \begin{claim}\label{claim:move_preservation}
        Let~$s$ be a truth assignment for~$S$. Assume~$|T| = |F| + 1$.
        If either $V \cap T$ is empty or flipping a vertex~$v \in V \cap T$ improves the cut induced
        by~$s$, then~$s$ is not locally optimal.
    \end{claim}

    \begin{claimproof}
        First, consider the case~$V \cap T = \emptyset$. Suppose it holds that~$q_v \in T$
        for all~$v \in V$. Consider some~$v \in V$ that can flip in the
        cut induced by~$s$. Note that~$v$ must exist, since one half of the cut is empty;
        therefore the induced cut cannot be locally optimal.
        A flip of $q_v$ unsatisfies all level 2 clauses, of the form
        $\nae(q_v, u)$, in which it occurs. This yields an improvement to the cost
        of at least $L$. To see this, note that at least one edge $uv$ incident to $v$ must have nonzero
        weight, and thus unsatisfying $C(q_v, u)$ contributes at least $+L$. 
        Flipping $q_v$ may in addition satisfy
        at most $2^5\cdot N$ level 3 clauses. Thus, flipping $q_v$ yields an improvement to the
        cost of at least
        \[
            L - 2^5\cdot N = 2^6 N - 2^5 N > 0.
        \]
        Now suppose that there exists some~$q_v \in F$. Consider a clause~$C_i(v, Q) = \nae(v, q_v, a_i)$
        for some~$a_i \in T$ (which exists due our choice
        of~$N$). If we flip~$a_i$, this clause becomes unsatisfied, which yields an
        improvement of $+1$ to the cost. Moreover, since $v \in F$ for all $v \in V$, there
        are no level 3 clauses that can become satisfied due to the flip of $a_i$. Hence,
        $s$ is not locally optimal.
        
        Now suppose that there exists some $v \in V \cap T$ that can
        flip in the cut induced by $s$. We then have
        \[
            \sum_{u \in N(v) \cap T} w(uv) - \sum_{u \in N(v) \cap F} w(uv) \geq 1.
        \]
        A flip of the level 1 variable $v$ satisfies the level 1 clauses $\nae(u, v)$ for $u \in N(v) \cap T$
        and unsatisfies those for $u \in N(v) \cap F$. Thus, this causes a change in the cost of at least~$M\cdot \Delta_{\min{}} \geq M$.

        Among the level 2 clauses,
        $v$ occurs in at most $5$ clauses, one clause $C(q_u, v)$ of weight $-Lw(uv)$ for each $u \in N(v)$.
        Thus, a flip of $v$ contributes at most
        $5 \cdot \Delta_{\max{}} \cdot L$ from the level 2 clauses.

        Finally, among the level 3 clauses, $v$ occurs in at most $2^5\cdot N$ clauses: for each level 3 variable,
        $v$ occurs in one clause for each subset of its neighborhood. Thus, the level 3 clauses contribute at most
        $2^5\cdot N$ when $v$ flips.

        Combining these estimates, a flip of $v$ increases the cost function by at least
        \[
            M - 5 \cdot \Delta_{\max{}} \cdot L - 2^5\cdot N > 0.
        \]
        Thus, $v$ can flip from \true{} to \false{} in $s$, proving the claim. 
    \end{claimproof}

    Let $s$ be a locally optimal truth assignment for $S$. 
    For the remainder of the proof, we assume that $|T(s)| = |F(s)| + 1$; this
    is w.l.o.g.,
    since the problem is symmetrical under inversion of the truth value of all variables
    simultaneously.
    Since $s$ is locally optimal, no level 1 variable in $T$ can be flipped
    for a positive gain. Hence, no vertex in the corresponding \textsc{Max Cut-5}
    solution can be flipped from $T \cap V$ to $F \cap V$ by \Cref{claim:move_preservation}.

    For any truth assignment with $|T| = |F|+1$, we must have some level 3 variable
    $a_i \in T$. We consider how flipping a level 3 variable $a_i \in T$ would affect the
    cost of $s$. Clearly, this cannot increase the cost,
    as then $s$ would not be locally optimal. The only options are then that the cost
    stays the same or decreases. The following claim excludes the latter possibility.
    
    We stress that we will not actually
    perform such a flip; the assumption that $s$ is locally optimal
    ensures that this is not a legal move.
    We will only consider the effect such a flip would have,
    in order to investigate the structure of locally optimal truth assignments.
    
    \begin{claim}\label{claim:auxiliary_flip}
        Flipping a level 3 variable in a locally optimal solution cannot cause any clauses with negative weight
        to become satisfied.
    \end{claim}

    \begin{claimproof}
        Consider a clause $C_i(v, Q) = \nae(v, q_v, a_i)$ constructed from
        $v \in V$ with negative weight, that becomes satisfied in a locally optimal
        solution $s$ due to flipping $a_i \in T$.
        Then we have~$v, q_v, a_i \in T$. (Recall the w.l.o.g.\ assumption that $|T| = |F| + 1$.)
      
        We know from
        \Cref{claim:move_preservation} that $v$ is unable to
        flip in the cut induced from $s$, as otherwise $s$ would not be locally optimal. Thus, 
        \begin{align}\label{eq:infeasible_flip}
            \sum_{u \in N(v) \cap T} w(uv) - \sum_{u \in N(v) \cap F} w(uv) \leq -\Delta_{\min{}}.
        \end{align}
        Now consider the change in cost that would occur due to a flip of $q_v$. Since
        $q_v$ is \true{}, flipping $q_v$ would satisfy the level 2 clauses
        $\nae(q_v, u)$ for $u \in N(v) \cap T$, and unsatisfy those for which
        $u \in N(v) \cap F$. By construction, this changes the cost by
        \[
            \left(\sum_{u \in N(v) \cap T} (-w(uv)) - \sum_{u \in N(v) \cap F} (-w(uv))\right)\cdot L 
            \geq \Delta_{\min{}} \cdot L,
        \]
        where the inequality follows from \Cref{eq:infeasible_flip}. 
        Flipping $q_v$ may also change the status of some level 3 clauses. For a fixed $a_i$,
        the variable $q_v$ occurs in all level 3 clauses generated from $v$. Thus, $q_v$ occurs in at most
        $2^5\cdot N$ level 3 clauses,
        once for each $a_i$ and for each subset of $N(v)$.
        The change in cost from flipping $q_v$ due to these clauses is then at most $2^5 N$
        in absolute value. Flipping $q_v$ then increases the cost by at least
        \[
            \Delta_{\min{}} \cdot L - 2^5\cdot N \geq L - 2^5\cdot N = 2^6\cdot N - 2^5 \cdot N > 0.
        \]
        We see that flipping $q_v$ would strictly increase the cost, contradicting local optimality
        of $s$. Thus, we must have $q_v \in F$, which then implies that flipping
        $a_i$ from \true{} to \false{} does not change the status of $C_i(v, Q)$.
        The claim follows.
    \end{claimproof}
    
    Suppose there exists some vertex $v \in V \cap F$ that can flip in
    the induced cut for a positive improvement in the cut weight. Let $Q = N(v) \cap F$.
    Consider the clause~$C_i(v, Q) = \nae(v, q_v, a_i)$ for some $a_i \in T$.
    Since flipping the vertex $v$ increases the weight of the cut, we know by construction that
    this clause has negative weight.
    
    We will show that we must have $q_v \in T$.
    Suppose towards a contradiction that $q_v \in F$.
    Then flipping $a_i \in T$ unsatisfies $C_i(v, Q)$, yielding a positive contribution to
    the cost. Moreover, by \Cref{claim:auxiliary_flip}, flipping $a_i$ cannot yield
    any negative contributions to the cost, and hence
    this flip is improving. This contradicts local optimality of $s$, and thus we must have~$q_v \in T$.

    Now consider a flip of $q_v$. This would satisfy the level 2 clauses $\nae(q_v, u)$
    for which~$u \in T$ and unsatisfy those for which $u \in F$.
    The cost would then change by at least
    \[
        \left(\sum_{u \in N(v) \cap F} w(uv) - \sum_{u \in N(v) \cap T} w(uv)\right)\cdot L - 2^5\cdot N
    \]
    as in the proof of \Cref{claim:auxiliary_flip}. But the quantity in brackets is exactly
    the change in cost that
    would be incurred by flipping $v$ in the induced cut, which was assumed
    improving; hence, flipping $q_v$ would increase
    the cost by at least $L - 2^5\cdot N = 2^6N - 2^5N > 0$, contradicting
    local optimality of $s$. Thus, no $v \in V \cap F$ can be eligible to flip in the induced cut.
    
    Together with \Cref{claim:move_preservation}, we can conclude that if $s$ is locally
    optimal, then there exists no vertex $v \in V$ that can flip in the cut induced by $s$.
    Thus, local optima for \halfposnaesat[3] map to a local
    optima of \maxcutd[5], satisfying the essential property of \PLS{}-reductions.
    Moreover, the reduction can be achieved in polynomial time, and so it is
    a valid \PLS{}-reduction.
    
    \proofsubparagraph{Tightness.} Let $x$ be an instance of
    \textsc{Distinct Max Cut-5} and let $(f, g)$ be the reduction given above.
    Let~$\R$ be the entire
    set of feasible truth assignments for $f(x)$. This obviously contains all local optima of
    $f(x)$, satisfying property 1 of tight reductions.
    Given some cut~$(A, B)$ in~$G$, we can construct~$s \in \R$ by
    setting the variables in~$A \cap S$ to \true{} and those in~$B \cap S$ to \false{}, assigning
    the level 2 variables a truth value arbitrarily, and assigning the level 3 variables
    such that~$|T| = |F| + 1$. This satisfies property~2 of tight reductions.
    For property~3, suppose the transition graph of the reduced instance~$f(x)$ contains
    a directed path~$P$ from~$s \in \R$ to~$s' \in \R$ with all internal vertices outside of~$\R$. Since~$\R$
    is the entire set of feasible truth assignments, this implies that~$P$ consists
    of the single arc~$ss'$. If~$s$ and~$s'$ differ in the flip of a
    level 2 or level 3 variable, then~$g(s, x) = g(s', x)$, and we are done.
    Suppose then that~$s$ and~$s'$ differ in
    the flip of a level 1 variable $v$. To conclude, we must show that the transition graph of $x$
    contains an arc from $g(s, x)$ to $g(s', x)$. Suppose to the contrary that no such arc exists.
    This implies that flipping the vertex $v$ in $g(s, x)$ \emph{reduces} the cut value by
    at least $1$. Thus, flipping the variable $v$ in $s$ reduces the cost of the truth
    assignment by at least $M - 5\Delta_{\max{}}L - 2^5N > 0$, which contradicts
    the existence of the arc~$ss'$.
\end{proof}

\fi

A reduction from \posnaesat[3] to \maxcut{} was provided by
Sch\"affer and Yannakakis \cite{schafferSimpleLocalSearch1991}.
Since \textsc{Max Cut} is equivalent to \textsc{Pos NAE 2-Sat},
we can use the same reduction to reduce from \halfposnaesat[3] to \halfposnaesat[2].
\iflong
We write the proof here for the sake of completeness.
\fi

\begin{lemma}[Sch\"affer and Yannakakis \cite{schafferSimpleLocalSearch1991}]
\label{lemma:reduce3to2}
    There exists a tight \PLS{}-reduction from \halfposnaesat[3] to \halfposnaesat[2].
\end{lemma}

\iflong

\begin{proof}
    Let $I$ be an instance of \halfposnaesat[3]. We assume without loss of generality that
    the weights of the clauses are all even. For every clause of size two in~$I$,
    we create an identical clause with the same weight in an instance $I'$ of
    \halfposnaesat[2]. For each clause of size three, we create three clauses:
    \[
        \nae(x_1, x_2, x_3) \to \nae(x_1, x_2), \nae(x_2, x_3), \nae(x_1, x_3).
    \]
    If the weight of the original clause is $W$, then the resulting clauses each get
    a weight of $W/2$. Note that this may result in several clauses consisting of the same
    pair of variables. For the purposes of this analysis, we consider each of these clauses to
    be distinct; however, one could equivalently merge these clauses into one clause with weight
    equal to the total weight of the merged clauses.
    
    Let $s$ be some feasible solution. If a clause $C$ of size three of weight $W$
    in $I$ is unsatisfied
    under $s$, then its resulting clauses in $I'$ are also unsatisfied. If $C$ is satisfied
    under $s$, then exactly two of its three resulting clauses are satisfied, yielding a contribution
    of $W$ to the objective function. Thus, the cost of a solution $s$ in $I$
    is exactly the cost of $s$ in $I'$. The reduction clearly preserves local optima,
    and is tight since the transition graphs of the two instances are identical.
\end{proof}

\fi

While our reductions so far have used negative-weight clauses in \halfposnaesat[k],
it may be of interest to have a \PLS{}-completeness result also when all clauses
have non-negative weight. 
\iflong
We briefly show that this can be done, as a corollary of \Cref{lemma:reduce3to2}. 
\fi

\begin{corollary}\label{cor:naesat_positive_weight}
    \halfposnaesat[2] is \PLS{}-complete even when all clauses have non-negative weight.
    More precisely, there exists a tight \PLS{}-reduction from \halfposnaesat[2] to \halfposnaesat[2] where all clauses have non-negative weight.
\end{corollary}

\iflong

\begin{proof}
    Given an instance of \halfposnaesat[2], we first merge any clauses consisting of the
    same pair of variables into a single clause of the same total weight. Next,
    we add clauses for each pair of variables
    $x, y$ that is not already part of some clause, giving each new clause a weight of zero.
    Note that this is equivalent to the original instance.
    Let the number of variables be $2n + 1$.
    Observe that for any truth assignment, exactly $n(n+1)$ clauses are satisfied.
    
    Next, we add to the weight of every clause a large constant $M > -\min_{\text{clauses $C$}} w(C)$,
    so that each clause has non-negative weight. The cost of a truth assignment is
    now \[ n(n+1)\cdot M + \sum_{\text{satisfied clauses $C$}}w(C),\] which is simply the weight of the same assignment
    in the original instance, plus a fixed term. Thus, the local optima of the reduced instance
    are identical to those of the original instance. Tightness follows since the transition graphs of
    the two instances are identical.
\end{proof}

\fi

Finally, we reduce from \halfposnaesat[2] to \minbisection{}.

\begin{lemma}\label{lemma:naetobisection}
    There exists a tight \PLS{}-reduction from \halfposnaesat[2]{} to both
    \maxbisection{} and \minbisection{}.
\end{lemma}

A reduction from \posnaesat[2]{} to \maxcut{} is given by
Sch\"affer and Yannakakis \cite{schafferSimpleLocalSearch1991}\iflong ; it is essentially
the reverse direction of the reduction given at the start of the proof of
\Cref{lemma:halfposnaesat}\fi. It is easy to see
that this reduction also works with our constraint on the number
of \true{} and \false{} variables, which yields a reduction to
\maxbisection{}.

\iflong

\begin{proof}
    By \Cref{cor:naesat_positive_weight}, we can assume that
    the clauses in the original instance all have non-negative weight, and thus this
    also holds for the edges in the resulting instance of \textsc{Odd Max Bisection}.
    
    To reduce to the minimization version, we first
    add all missing edges to the instance with a weight of zero, which does not
    affect the local optima.
    We then multiply the weight of each edge by $-1$, and
    add a large constant to the weight of each edge; ``large'' here meaning
    large enough to make all weights non-negative. Since the number of edges in
    the cut is the same for every feasible solution, this only changes the cost function
    by a fixed term.
    
    We take the resulting
    instance as an instance of \minbisection{}. The local maxima of
    the \maxbisection{} instance become local minima, and so locally
    optimal solutions are preserved, concluding the proof. For tightness, we observe
    that the transition graphs of the instances are identical.
\end{proof}

\fi

\section{Reduction to Clustering Problems}
\label{sec:reduce to clustering}

Armed with the \PLS{}-completeness of \minbisection{} (see \Cref{lemma:naetobisection}), we now proceed to prove hardness
of the Euclidean clustering problems of interest.

\iflong
\subsection{\boldmath \texorpdfstring{$k$}{k}-Means}
\fi

We provide a tight \PLS-reduction from \minbisection{} to \khartigan{}. 
This is done in three steps (see \Cref{fig:reduction_graph}).
First, we show \PLS{}-completeness of \densestcut{}.
The construction of the proof of our \PLS{}-completeness of \densestcut{} is rather simple (we only add a large set of isolated edges), but the analysis of the correctness is quite technical.
Second, we show \PLS{}-\iflong completeness \fi of \hartigan{} by slightly adapting an \NP-hardness reduction of \textsc{$2$-Means}~\cite{ADHP09}.
Finally, we extend this result to \khartigan{}.

Now, we show \PLS{}-completeness of \densestcut{}.
We impose the additional~constraint that there are no isolated vertices in the reduced instance.
This is a technical condition which is utilized in \Cref{lemma:densestcut_to_hartigan} for the \PLS-hardness of \khartigan{}.

\begin{lemma}\label{lemma:bisection_to_densestcut}
    There exists a tight \PLS{}-reduction from \maxbisection{} to \densestcut{}
    without isolated vertices.
\end{lemma}

\proofsubparagraph{Sketch of Proof for \Cref{lemma:bisection_to_densestcut}.}
Similar to the obstacle to prove \Cref{lemma:halfposnaesat}, it is not difficult to show hardness for \textsc{Densest Cut} for swapping two vertices: the main technical hurdle arises when considering \flip, the most fundamental neighborhood.
The construction of our \densestcut{} instance is based on the simple idea of adding $n^4$~isolated edges to~$G$, 
where~$(G,w)$ is an instance of~\maxbisection{}; see \Cref{fig max cut hardness}.
The weight~$M$ of these newly added edges has to be chosen carefully:  On the one hand,~$M$ has to be larger than the maximum weight assigned by~$w$, 
but on the other hand~$M$ cannot be too large to ensure a specific structure of locally optimal solutions.

In a first step, we show that each locally optimal solution contains all these additional edges.
Intuitively, this holds due to the large weight of these edges.
This implies that each locally optimal solution has to contain at least~$n^4$ vertices on both parts of the partition.
Hence, for such a partition, flipping a vertex from one part of the partition to the other part only slightly changes the denominator in the objective function.
We use this observation to show in a second step that for each locally optimal solution~$(A,B)$, the number of vertices of~$V(G)$ in~$A$ and the number of vertices of~$V(G)$ in~$B$ differ by exactly one. 
This implies that~$(A\cap V(G),B\cap V(G))$ is a valid solution for the~\maxbisection{}-instance.

To show this statement, we exploit the special structure of the \maxbisection{}-instance. 
More precisely, if both parts of the cut differ in their size by more than one, then we can flip an arbitrary vertex from the larger part to the smaller part and increase the total weight of the resulting partition.
This is true since we will ensure that~$G$ is a complete graph where an edge set of larger cardinality is always guaranteed to have a larger total weight.
Based on these observations, we then show the second step, that is, that for each locally optimal solution~$(A,B)$, the number of vertices of~$V(G)$ in~$A$ and the number of vertices of~$V(G)$ in~$B$ differ by exactly one. 

Finally, to show that each locally optimal solution~$(A,B)$ for the~\densestcut{}-instance maps back to a locally optimal solution~$(A\cap V(G),B\cap V(G))$ for the~\maxbisection{}-instance, we exploit that the denominator of the objective function for~\densestcut{} stays the same when flipping a single vertex from the larger to the smaller part of the partition.  
Hence, if for some vertex~$v\in V(G)\cap B$, flipping~$v$ from~$B$ to~$A$ yields a better solution for the~\maxbisection{}-instance than~$(A\cap V(G),B\cap V(G))$, then we can conclude that flipping~$v$ from~$B$ to~$A$ yields a better solution for the~\densestcut{}-instance than~$(A,B)$.
This then implies that the reduction is correct. \hfill$\blacksquare$


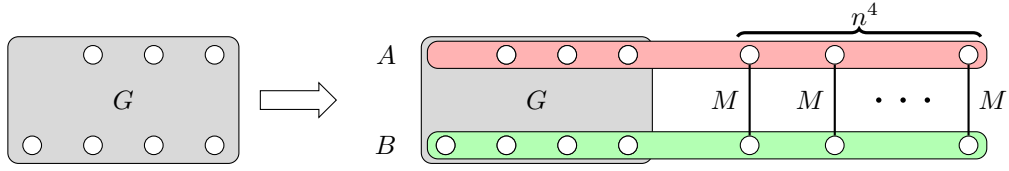
\begin{figure}[t]
\begin{center}
\begin{tikzpicture}[xscale=.8,yscale=1]
\tikzstyle{knoten}=[circle,fill=white,draw=black,minimum size=7pt,inner sep=0pt]
\tikzstyle{blocked}=[rectangle,fill=white,draw=black,minimum size=7pt,inner sep=0pt]
\tikzstyle{bez}=[inner sep=0pt]

\pgfdeclarelayer{background}
\pgfdeclarelayer{foreground}
\pgfsetlayers{background,main,foreground}

    \begin{scope}[xshift=-6.8cm]
        \draw[rounded corners, fill=gray!30] (-2.4, 2.4) rectangle (1.4, -0.4) {};
		\node (v) at (-.5,1) {$G$};
		\node[knoten] (v) at (-2,0) {};
		\node[knoten] (vu) at ($(v) + (1,0)$) {};
		\node[knoten] (vw) at ($(vu) + (1,0)$) {};
		\node[knoten] (vx) at ($(vw) + (1,0)$) {};

        \node[knoten] (wu) at ($(vu) + (0,2)$) {};
		\node[knoten] (ww) at ($(vw) + (0,2)$) {};
		\node[knoten] (wx) at ($(vx) + (0,2)$) {};

    \node[single arrow, draw=black,
      minimum width = 10pt, single arrow head extend=3pt,
      minimum height=10mm] at ($(vx) + (1.3,1)$) {};

    \end{scope}

        \begin{pgfonlayer}{background}
        \draw[rounded corners, fill=gray!30] (-2.4, 2.4) rectangle (1.4, -0.4) {};
  \end{pgfonlayer}

		\node (v) at (-.5,1) {$G$};
		\node[knoten] (v) at (-2,0) {};
		\node[knoten] (vu) at ($(v) + (1,0)$) {};
		\node[knoten] (vw) at ($(vu) + (1,0)$) {};
		\node[knoten] (vx) at ($(vw) + (1,0)$) {};

        \node[knoten] (wu) at ($(vu) + (0,2)$) {};
		\node[knoten] (ww) at ($(vw) + (0,2)$) {};
		\node[knoten] (wx) at ($(vx) + (0,2)$) {};

		\node[knoten] (v) at (-2,0) {};
		\node[knoten] (vu) at ($(v) + (1,0)$) {};
		\node[knoten] (vw) at ($(vu) + (1,0)$) {};
		\node[knoten] (vx) at ($(vw) + (1,0)$) {};

        \node[knoten] (wu) at ($(vu) + (0,2)$) {};
		\node[knoten] (ww) at ($(vw) + (0,2)$) {};
		\node[knoten] (wx) at ($(vx) + (0,2)$) {};

		\node[knoten] (M1) at ($(vx) + (2,0)$) {};
		\node[knoten] (O1) at ($(M1) + (0,2)$) {};
		
		\node[knoten] (M2) at ($(M1) + (1.4,0)$) {};
		\node[knoten] (O2) at ($(M2) + (0,2)$) {};
		
		
		\node[] (dots) at ($(M2) + (1.2,1)$) {\huge $\dots$};
		
		\node[knoten] (M5) at ($(M2) + (2.2,0)$) {};
		\node[knoten] (O5) at ($(M5) + (0,2)$) {};

		\draw[thick] (M1) -- (O1) node [midway, left] {$M$};
		\draw[thick] (M2) -- (O2) node [midway, left] {$M$};
		\draw[thick] (M5) -- (O5) node [midway, right] {$M$};
\begin{scope}[yshift=-1cm]
		\node (a)[label=left:{$B$}] at ($(v) + (-.5,0)$) {};
		\node (b)[label=left:{$A$}] at ($(a) + (0,2)$) {};

\draw[very thick,decorate,decoration={brace,amplitude=3pt}] 
    ($(O1) + (-.2,0.4)$) coordinate (t_k_unten) -- ($(O5) + (.2,0.4)$) coordinate (t_k_opt_unten); 
\node at ($0.5*(O1)+0.5*(O5) + (0.1,.9)$) {$n^4$};

        \begin{pgfonlayer}{background}
    \draw[rounded corners, fill=red!30] ($(v) + (-.3,1.7)$) rectangle ($(O5) + (.3,0.3)$) {};
    \draw[rounded corners, fill=green!30] ($(v) + (-.3,-.3)$) rectangle ($(M5) + (.3,.3)$) {};
  \end{pgfonlayer}
\end{scope} 
		
		\end{tikzpicture}
\end{center}
\caption{Schematic overview of the reduction used in the proof of \Cref{lemma:bisection_to_densestcut}.
On the left side we have an instance of \maxbisection{} and on the right side we have the corresponding instance of \densestcut{}. 
The edges inside of~$G$ together with their weights are not depicted but are identical in both instances.
Let~$(A,B)$ be the partition corresponding to some locally optimal solution of the \densestcut{} instance.
Then,~$A$ contains exactly one endpoint of each of the $n^4$~isolated edges and~$|\, |A\cap V(G)|-|B\cap V(G)|\, |=1$.}
\label{fig max cut hardness}
\end{figure}

\iflong

\begin{proof}[Proof of \Cref{lemma:bisection_to_densestcut}]
    Before we present our construction for the instance of \densestcut{}, we present a construction for an auxiliary instance  of \maxbisection{} such that all edge-weights are close to each other.
    This property is important to verify that locally optimal solutions can be mapped.

\proofsubparagraph{Construction of an auxiliary instance.}
%
    Let~$\widehat{G}=(V,\widehat{E})$ be an instance of~\maxbisection{} with weights $\widehat{w}$ on the edges and let~$n \coloneqq |V|$.
    We now construct an instance~$G=(V,E)$ of~\maxbisection{} with weight vector~$w$ such that both instances share the same locally optimal solutions and the minimal weight and the maximal weight assigned by~$w$ are \emph{close} to each other.
    Note that~$G$ and~$\widehat{G}$ share the same vertex set~$V$.
    We define~$G$ to be the complete graph on the vertex set~$V$.
    Let~$\widehat{w}_{\max} \coloneqq \max_{e\in \widehat{E}} w(e)$.
    We define the weights $w \in \mathds{Z}_{\geq 0}^E$ as follows:
    For each edge~$uv\in \widehat{E}$, we set~$w(uv) \coloneqq n^9 \cdot \widehat{w}_{\max} + w(uv)$ and for each~$uv\in E\setminus \widehat{E}$, we set~$w(uv) \coloneqq n^9 \cdot \widehat{w}_{\max}$.
    Now, let~$w_{\min} \coloneqq \min_{e\in E} w(e)$ and~$w_{\max} \coloneqq \max_{e\in E} w(e)$. 
    Note that based on the defined weights, we obtain the following property:
    \begin{align}
        \left(1 - \frac{1}{n^9}\right) \cdot w_{\max} \leq w_{\min} \leq \left(1 + \frac{1}{n^9}\right) \cdot w_{\max}.\label{lab-edge-weights}
    \end{align}
    
\proofsubparagraph{Correctness of the auxiliary instance.}    We now show that both~$G$ and~$\widehat{G}$ share the same locally optimal solutions.

    \begin{claim}\label{claim-g-and-ghat-equivalent}
        $G$ and~$\widehat{G}$ share the same locally optimal solutions.
    \end{claim}

\begin{claimproof}
        Recall that each solution~$(A,B)$ for~$G$ or~$\widehat{G}$ is a partition of~$V$ such that~$|\,|A| - |B|\,| = 1$.
        Let~$(A,B)$ be a solution for~$G$ or~$\widehat{G}$.
        By the above, $|A|\cdot |B| = \frac{n+1}{2} \cdot \frac{n-1}{2}$.
        Hence, $$w(A,B) = |A|\cdot |B| \cdot n^9 \cdot \widehat{w}_{\max} + \widehat{w}(A,B) = \frac{n+1}{2} \cdot \frac{n-1}{2} \cdot n^9 \cdot \widehat{w}_{\max} + \widehat{w}(A,B).$$ 
        This implies that a solution~$(A',B')$ for~$G$ or~$\widehat{G}$ improves over~$(A,B)$ with respect to~$w$ if and only if~$(A',B')$ improves over~$(A,B)$ with respect to~$\widehat{w}$.
        Consequently, $G$ and~$\widehat{G}$ share the same locally optimal solutions and the reduction is tight since the transition graphs of both instances are identical.
\end{claimproof}
        
\proofsubparagraph{Implications.} The bounds between~$w_{\min}$ and~$w_{\max}$ (see \Cref{lab-edge-weights}) are important to ensure that
    one can always improve over unbalanced partition by at least~$w_{\min}$ by flipping only a single vertex to the other part of the partition if the cardinality constraint is omitted.
    
    \begin{claim}\label{claim inbalance implies large improvement}
        Let~$X \subseteq V$ and let~$Y \coloneqq V \setminus X$.
        If~$|X| < |Y|-1$, then for each vertex~$v\in Y$, $w(X\cup \{v\},Y\setminus \{v\}) - w(X,Y) > w_{\min}$.
    \end{claim}
    \begin{claimproof}
        Let~$v$ be an arbitrary vertex of~$Y$ and let~$X' \coloneqq X \cup \{v\}$ and~$Y' \coloneqq Y \setminus \{v\}$.
        Since~$|X| + |Y| = n$ is odd, $|X| < |Y|-1$ implies that~$|X| \leq |Y|-3$.
        Moreover, since~$G$ is a complete graph, 
        \begin{align*}
        |E(X',Y')| - |E(X,Y)| = |X'| \cdot |Y'| - |X| \cdot |Y| = |Y| - |X| - 1 \geq 2.
        \end{align*}
        This implies that:
        \begin{align*}
            w(X',Y') - w(X,Y) &\geq |E(X',Y')|\cdot w_{\min} -|E(X,Y)|\cdot w_{\max} \\
            &\geq (|E(X,Y)| + 2)\cdot w_{\min} -|E(X,Y)|\cdot w_{\max} \\
            &= 2\cdot w_{\min} - |E(X,Y)| \cdot (w_{\max} - w_{\min}) \\
            &\geq 2\cdot w_{\min} - |E(X,Y)| \cdot \left(\left(1+\frac{1}{n^9-1}\right) \cdot w_{\min} - w_{\min}\right) \\
            &= 2\cdot w_{\min} - |E(X,Y)| \cdot \frac{1}{n^9-1} \cdot w_{\min} \\
            &\geq 2\cdot w_{\min} - n^2 \cdot \frac{1}{n^9-1} \cdot w_{\min} \\
            &> 2\cdot w_{\min} -  w_{\min} = w_{\min}.
        \end{align*}
    \end{claimproof}  

\proofsubparagraph{Construction of the \densestcut{} instance.}    
    We now construct the instance~$G'\coloneqq (V',E')$ of~\densestcut{} as follows; see \Cref{fig max cut hardness}.
    To obtain~$G'$, we add a matching of size~$n^4$ with edge set~$E_M$ to~$G$.
    Hence, $V'\setminus V$ denotes the endpoints of the matching edges of~$E_M$.
    Finally, we obtain the weights on the edges of~$E'$ by extending the weight function~$w$:
    We define the weight of each edge of~$E'\setminus E = E_M$ as~$M \coloneqq n \cdot w_{\max}$.
    
    \proofsubparagraph{Correctness.}
    Next, we show that we can map each locally optimal solution of~$G'$ in polynomial time to a locally optimal solution of~$G$ (and thus also for~$\widehat{G}$, according to \Cref{claim-g-and-ghat-equivalent}).
    To this end, we first show that each locally optimal solution fulfills two helpful properties.
    First, we show that in a locally optimal solution all edges of~$E_M$ are contained in~$E(A,B)$.

    \begin{claim}\label{claim matching}
        Let~$(A,B)$ be a partition of~$G'$ such that at least one edge~$uv\in E_M$ is not contained in~$E(A,B)$.
        Then~$(A,B)$ is not locally optimal.
    \end{claim}
    
    \begin{claimproof}
        Since~$|V'|$ is odd, we can assume without loss of generality that~$|A| \leq |B|-1$.
        Essentially, the idea is to flip one endpoint of the edge~$uv$ to the other part of the partition.
        We show that either this is improving, or one can obtain a better partition by flipping some vertex of~$V$.
        
        First, we show that flipping one endpoint of the edge~$uv$ to the other part of the partition is improving if~$A$ contains both~$u$ and~$v$.
        More precisely, we show that~$(A\setminus \{v\},B \cup \{v\})$ improves over~$(A,B)$.
        Note that $$w(A\setminus \{v\},B \cup \{v\}) = w(A,B) + w(uv) > w(A,B).$$
        Moreover, since~$|A| \leq |B|-1$, we have $$|A\setminus \{v\}|\cdot |B \cup \{v\}| = |A| \cdot |B| - |B| + |A| - 1 < |A| \cdot |B|.$$
        Hence, we directly obtain~$\frac{w(A\setminus \{v\},B \cup \{v\})}{(|A|-1)\cdot (|B|+1)} > \frac{w(A,B)}{|A|\cdot |B|}$ which implies that~$(A,B)$ is not locally optimal.
        Hence, in the following, we may thus assume that~$B$ contains both~$u$ and~$v$.
        Note that~$|A|\geq 1$ by definition of feasible solutions of the problem.
        
        To prove the claim, we distinguish three cases. 
        In each of them, we analyze the (potentially non-positive) improvement of the partition~$(A' \coloneqq A \cup \{v\}, B' \coloneqq B \setminus \{v\})$ over~$(A,B)$,
        which is given by~$\frac{w(A', B')}{|A'|\cdot |B'|} - \frac{w(A,B)}{|A|\cdot |B|}$.
        To determine whether this improvement is positive, we equivalently
        analyze the quantity~$w(A', B') - \frac{|A'|\cdot |B'|}{|A|\cdot |B|} \cdot w(A,B)$.
        Since~$E(A',B') = E(A,B) \cup \{uv\}$, we obtain that~$w(A',B') = w(A,B) + M$.
        Hence, $w(A', B') - \frac{|A'|\cdot |B'|}{|A|\cdot |B|} \cdot w(A,B)$ can be rewritten as:
        \begin{align*}
            w(A', B') - \frac{|A'| \cdot |B'|}{|A|\cdot |B|}\cdot w(A,B) &= w(A , B) + M - \left(1 + \frac{|B| - |A| - 1}{|A|\cdot |B|}\right)\cdot w(A,B) \\
            &= M - \frac{|B| - |A| - 1}{|A|\cdot |B|}\cdot w(A,B).
        \end{align*}
        
        To show that~$(A,B)$ is not locally optimal, we are finally ready to distinguish three cases based on the size of~$A$ and the intersection of~$A$ and~$V$:
        Either~$A\cap V=\emptyset$ (Case~$2$), or~$A\cap V\ne \emptyset$.
        If~$A\cap V\ne \emptyset$, we additionally consider the size of~$A$.
        Note that this implies~$|A|\ge 1$. 
        More precisely, we consider~$|A|\ge 2$ (Case~$3$), and~$|A|=1$ (Case~$1$).
        Note that~$|A|=1$ and~$A\cap V\ne\emptyset$ implies~$A\subseteq V$.

        \begin{description}
            \item[Case 1:~$|A| = 1$ and~$A \subseteq V$.]
                We show that~$(A',B')$ improves over~$(A,B)$, that is, we show that~$M - \frac{|B| - |A| - 1}{|A|\cdot |B|}\cdot w(A,B) > 0$.
                By the assumption of the case, we have~$|A| = 1$, $|B|-|A|-1 = |V'| - 2$, and~$|A| \cdot |B| = |V'|-1$.
                Hence, it suffices to show that~$M > w(A,B)$.
                
                Since the unique vertex of~$A$ is from~$V$, we see that
                $E(A,B)$ contains the~$n-1$ edges between the unique vertex of~$A$ and each vertex of~$V\setminus A$.
                Since each such edge has weight at most~$w_{\max}$, $w(A,B) < n\cdot w_{\max} = M$.
                Hence, $M > w(A,B)$, which implies that flipping vertex~$v$ from~$B$ to~$A$ is improving.
            
            \item[Case 2:~$A \cap V  = \emptyset$.]
                Again, we show that~$(A',B')$ improves over~$(A,B)$, that is, we show that~$M - \frac{|B| - |A| - 1}{|A|\cdot |B|}\cdot w(A,B) > 0$.
                
                Note that since~$A$ contains no vertex of~$V$, we have $E(A,B) \subseteq E_M$.
                Moreover, since each vertex of~$V'\setminus V$ is incident with only one edge, $|E(A,B)| \leq \min(|A|,|B|) = |A|$.
                These two facts imply that~$w(A,B) \leq |E(A,B)| \cdot M = |A| \cdot M$.
                Consequently: 
                \begin{align*}
                    M - \frac{|B|-|A|-1}{|A|\cdot |B|} \cdot w(A,B) &\geq M - \frac{|B|-|A|-1}{|A|\cdot |B|} \cdot |A| \cdot M\\
                        &= M - \frac{|B|-|A|-1}{|B|} \cdot M >  0.
                \end{align*}
                Hence, flipping vertex~$v$ from~$B$ to~$A$ is improving.
        
            \item[Case 3:~$A\cap V \neq \emptyset$, and~$|A| > 1$.]
                We show that one can obtain a better solution by flipping the vertex~$v$ from~$B$ to~$A$ or by flipping any vertex~$x$ of~$V\cap A$ from~$A$ to~$B$.
                
                If~$(A',B')$ improves over~$(A,B)$, then~$(A,B)$ is not locally optimal and the statement holds.
                Hence, we assume in the following that this is not the case, that is, $\frac{w(A',B')}{|A'|\cdot |B'|} - \frac{w(A,B)}{|A|\cdot |B|} \leq 0$.
                This implies that~$M - \frac{|B| - |A| - 1}{|A|\cdot |B|}\cdot w(A,B) \leq 0$, which is equivalent to
                \begin{align}\label{eq flip the other way}
                    \frac{|B| - |A| - 1}{|A|\cdot |B|}\cdot w(A,B) \geq M.
                \end{align}
                We show that this implies that flipping an arbitrary vertex~$x$ of~$V\cap A$ from~$A$ to~$B$ is improving.
                Let~$A'' \coloneqq A\setminus \{x\}$ and~$B'' \coloneqq B \cup \{x\}$.
                We show that~$\frac{w(A'', B'')}{(|A|-1)\cdot (|B|+1)} - \frac{w(A,B)}{|A|\cdot |B|} > 0$.
                It is equivalent to show that~$w(A'', B'') - \frac{(|A|-1)\cdot (|B|+1)}{|A|\cdot |B|} \cdot w(A,B) > 0$.
                
                Note that $$E(A,B) \setminus  E(A'',B'') = E(\{x\}, B) = E(\{x\}, B \cap V).$$
                Hence, $$w(A'',B'') \geq w(A,B) - (n-1)\cdot w_{\max} > w(A,B) - M.$$
                Consequently, 
                \begin{align*}
                    w(A'', B'') &- \frac{(|A|-1)\cdot (|B|+1)}{|A|\cdot |B|} \cdot w(A,B) \\
                        &>  w(A,B) - M - \left( 1+ \frac{|A|-|B|-1}{|A|\cdot |B|}\right) \cdot w(A,B) \\
                        &= -M - \frac{|A|-|B|-1}{|A|\cdot |B|} \cdot w(A,B) \\
                        &= \frac{|B|-|A|+1}{|A|\cdot |B|} \cdot w(A,B) - M \\
                        &\stackrel{\text{(\ref{eq flip the other way})}}{\geq} \frac{|B|-|A|+1}{|A|\cdot |B|} \cdot w(A,B) - \frac{|B|-|A|-1}{|A|\cdot |B|} \cdot w(A,B) \\
                        &= \frac{2}{|A|\cdot |B|} \cdot w(A,B) > 0.
                \end{align*}
                Hence, $(A,B)$ is not locally optimal.
        \end{description}
    \end{claimproof}

    In the following, we thus assume that for each partition~$(A,B)$ of~$G'$, the cut
    $E(A,B)$ contains all edges of~$E_M$.
    Hence, both~$A$ and~$B$ have size at least~$n^4$.
    We exploit the fact that both~$A$ and~$B$ are large to verify that in a locally optimal solution of~$G'$ 
    the vertices of~$G$ are balanced for both partite sets.
    
    \begin{claim}\label{claim balance}
        Let~$(A,B)$ be a partition of~$G'$ with~$E_M \subseteq E(A,B)$, such that~$|B| - |A| > 1$.
        Then, $(A,B)$ is not locally optimal.
        More precisely, for each vertex~$v\in V \cap B$, the partition~$(A\cup \{v\},B\setminus \{v\})$ improves over the partition~$(A,B)$.
    \end{claim} 
    
    \begin{claimproof}
        Let~$X\coloneqq A \cap V$ and let~$Y \coloneqq B \cap V$.
        Note that since~$E(A,B)$ contains all edges of~$E_M$, we have $|A\setminus V| = |B\setminus V| = n^4$.
        Hence, $|B| - |A| > 1$ implies that~$|Y| - |X| > 1$.
        Let~$v$ be an arbitrary vertex of~$Y$.
        We show that flipping vertex~$v$ from~$B$ to~$A$ is improving. 
        
        To this end, let~$X' \coloneqq X \cup \{v\}$ and~$Y' \coloneqq Y \setminus \{v\}$.
        We show that~$\frac{w(A\cup \{v\}, B \setminus \{v\})}{(|A| + 1) \cdot (|B| - 1)} - \frac{w(A,B)}{|A| \cdot |B|} > 0$.
        It is equivalent to show that~$$w(A\cup \{v\}, B \setminus \{v\}) - \frac{(|A| + 1) \cdot (|B| - 1)}{|A| \cdot |B|} \cdot w(A,B) > 0.$$
        Note that~$$w(A,B) = w(A\setminus V,B\setminus V) + w(X,Y) = n^4 \cdot M + w(X,Y).$$
        %
        Moreover, since~$|X| < |Y| - 1$, \Cref{claim inbalance implies large improvement} implies that~$w(X',Y') > w(X,Y) + w_{\min}$.
        This implies that 
        \begin{align*}
            w(A\cup \{v\}, B \setminus \{v\}) &= w(A\setminus V,B\setminus V) + w(X',Y') \\
            &= n^4 \cdot M + w(X',Y') > n^4 \cdot M + w(X,Y) + w_{\min}.
        \end{align*}
        
        Since~$|A| = n^4 + |X|$ and~$|B| = n^4 + |Y|$, we thus derive:
        \begin{align*}
            w(A &\cup \{v\}, B \setminus \{v\}) - \frac{(|A| + 1) \cdot (|B| - 1)}{|A| \cdot |B|} \cdot w(A,B)\\
            & > n^4 \cdot M + w(X,Y) + w_{\min} - \frac{(|A| + 1) \cdot (|B| - 1)}{|A| \cdot |B|}\cdot (n^4 \cdot M + w(X,Y)) \\
            &= n^4 \cdot M + w(X,Y) + w_{\min} - \left(1 + \frac{|B| - |A| - 1}{|A| \cdot |B|})\cdot (n^4 \cdot M + w(X,Y)\right)\\
            &= w_{\min} - \frac{|B| - |A| - 1}{|A| \cdot |B|}\cdot (n^4 \cdot M + w(X,Y))\\
            &= w_{\min} - \frac{|Y| - |X| - 1}{(n^4 + |X|) \cdot (n^4 + |Y|)}\cdot (n^4 \cdot M + w(X,Y)) \\
            &\geq  w_{\min} - \frac{n}{n^8}\cdot (n^4 \cdot M + w(X,Y)) \\
            &>   w_{\min} - \frac{1}{n^7}\cdot (n^4 \cdot M + n^2 \cdot w_{\max}).
        \end{align*}
        Here, the inequality~$n^2\cdot w_{\max}\ge w(X,Y)$ follows from the fact that there are at most $n^2$~edges between~$X$ and~$Y$, each with weight at most~$w_{\max}$.
        Now, since~$M = n\cdot w_{\max}$ and~$w_{\max} \leq \frac{n^9}{n^9-1} \cdot w_{\min} \leq 2 \cdot w_{\min}$ (see \Cref{lab-edge-weights}), we conclude:
        \begin{align*}
            w_{\min} - \frac{1}{n^7}\cdot (n^4 \cdot M + n^2 \cdot w_{\max}) &= w_{\min} - \frac{1}{n^7}\cdot (n^5 \cdot w_{\max} + n^2 \cdot w_{\max}) \\
            &= w_{\min} - \frac{1}{n^7}\cdot w_{\max} \cdot (n^5 + n^2) \\
            &\geq  w_{\min} - \frac{1}{n^7}\cdot w_{\min} \cdot 2 \cdot (n^5 + n^2) \\
            &\geq  w_{\min} - \frac{1}{n^7}\cdot w_{\min} \cdot n^6  > 0.
        \end{align*}
        Hence, flipping an arbitrary vertex of~$V\cap Y$ from~$B$ to~$A$ provides a better solution.
    \end{claimproof}
    
    As a consequence, \Cref{claim matching,claim balance} imply that each locally optimal solution~$(A,B)$ for~$G'$ fulfills~$E_M \subseteq E(A,B)$ and~$|\,|A\cap V| - |B\cap V|\,| = 1$.
    Hence, $(A\cap V, B\cap V)$ is a feasible solution for the instance~$G$ of~\maxbisection{}.
    
    \proofsubparagraph{Mapping of locally optimal solutions.}
    We are now ready to show that for each solution~$(A,B)$ for~$G'$ with~$E_M \subseteq E(A,B)$ and~$|\,|A\cap V| - |B\cap V|\,| = 1$, $(X\coloneqq A \cap V,Y\coloneqq B \cap V)$ is a locally optimal solution for~$G$ if~$(A,B)$ is a locally optimal solution for~$G'$.
    We show this statement by contraposition, that is, we show that if~$(X,Y)$ is not locally optimal, then~$(A,B)$ is not locally optimal.
    Assume without loss of generality that~$|X| = |Y|-1$.
    Since~$(X,Y)$ is not a locally optimal solution for~$G$, there exists a vertex~$v\in Y$ such that flipping~$v$ from~$Y$ to~$X$ provides a better solution, that is, $w(X\cup \{v\}, Y\setminus \{v\}) - w(X, Y)> 0$.
    Consider the partition~$(A',B')$ with~$A' \coloneqq A \cup \{v\}$ and~$B' \coloneqq B \setminus \{v\}$.
    Since~$|B| = |A| + 1$, this implies that~$|A| \cdot |B| = |A'| \cdot |B'|$.
    Hence, the improvement of~$(A',B')$ over~$(A,B)$ is:
    \begin{align*}
        \frac{w(A',B')}{|A'|\cdot|B'|} - \frac{w(A,B)}{|A|\cdot|B|} &= \frac{w(A',B') - w(A,B)}{|A|\cdot|B|}\\ 
            & = \frac{n^4 \cdot M + w(X\cup \{v\}, Y\setminus \{v\}) - n^4 \cdot M - w(X,Y)}{|A|\cdot|B|} \\
            &= \frac{w(X\cup \{v\}, Y\setminus \{v\}) - w(X,Y)}{|A|\cdot|B|}.
    \end{align*}
    This improvement is positive by the assumption that~$w(X\cup \{v\}, Y\setminus \{v\}) - w(X, Y)> 0$.
    Hence, $(A,B)$ is not a locally optimal solution for~$G'$.
    
    Consequently, for each locally optimal solution~$(A,B)$ for~$G'$, $(A\cap V, B\cap V)$ is a locally optimal solution for~$G$.   

    \proofsubparagraph{Tightness.}
    We obtain tightness of the reduction by setting~$\R$ to be the set of all partitions~$(A,B)$ of~$G'$ that
    fulfill~$E_M \subseteq E(A,B)$ and~$|\,|A| - |B|\,| = 1$.
    Note that by the above, this implies that~$\R$ contains all local optima for~$G'$.
    Moreover, for each bisection~$(X,Y)$ of~$V$, $\R$ contains at least one solution~$(A,B)$ with~$A\cap V = X$ and~$B\cap V = Y$.
    Hence, the first two properties of a tight~\PLS-reduction follow.
    
    It thus remains to show the final property of a tight~\PLS-reduction. 
   To this end, we first show that for each solution in~$\R$ all its improving neighbors are also contained in~$\R$.
    Let~$(A,B) \in \R$ be a solution for~$G'$ such that~$|B| = |A| + 1$.
    We consider the possible flip-neighbors of~$(A,B)$ that are not in~$\R$.
    These are the partitions~$(A',B')$ for which one of the following holds:
\begin{enumerate}[(a)]
   \item  there is some endpoint~$v$ of some edge of~$E_M$ such that~$v\in B$, $A' \coloneqq A \cup \{v\}$, and~$B' \coloneqq B\setminus \{v\}$;
    \item there is some endpoint~$v$ of some edge of~$E_M$ such that~$v\in A$, $A' \coloneqq A \setminus \{v\}$, and~$B' \coloneqq B\cup \{v\}$;
    or
    \item there is some endpoint~$v\in V \cap A$, such that~$A' \coloneqq A \setminus \{v\}$ and~$B'\coloneqq B\cup \{v\}$.
\end{enumerate}
    
    Note that for~(a): $|A'| \cdot|B'| = |A| \cdot |B|$ and thus, the denominator in the objective functions of both
    solutions stays the same.
    On the other hand, $E(A',B')$ is equal to~$E(A,B)$ minus the weight~$M>0$ of the unique edge incident with~$v$.
    Consequently, $(A',B')$ is not an improving flip-neighbor of~$(A,B)$.
    It thus remains to show that for~(b) and~(c), $(A',B')$ is also not an improving flip-neighbor of~$(A,B)$.
    
    In cases~(b) and~(c) the denominator in the objective function is equal.
    For~(c) the numerator is strictly larger than the numerator of~(b).
    This is the case, since in~(b), the cut loses an edge of weight~$M > n \cdot w_{\max}>0$, whereas in~(c), the
    cut loses the weight of at most~$n$ edges of weight at most~$w_{\max}$ each.
    Hence, it suffices to show that for~(c), $(A',B')$ is not improving over~$(A,B)$.
    
    Due to~\Cref{claim balance}, $(A',B')$ can be improved by flipping any vertex of~$V$ from~$B'$ to~$A'$.
    This includes flipping the vertex~$v$ back from~$B'$ to~$A'$.
    Consequently~$(A',B')$ is not improving over~$(A,B)$.
    This implies that for each solution of~$\R$ all improving flip-neighbor are in~$\R$ as well.
    
    Concluding, to prove the third property of a tight~\PLS-reduction, it suffices to show that for each arc~$(s,s')$ in the transition graph for~$G'$ with~$s = (A,B)\in \R$ and~$s'=(A',B')\in \R$, there is an arc between~$(A\cap V,B\cap V)$ and~$(A'\cap V, B'\cap V)$ in the transition graph for~$G$.
    This follows from the same argumentation as the above proof that each solution~$(A,B)$ is locally optimal for~$G'$
    if and only if~$(A\cap V, B\cap V)$ is locally optimal for~$G$.
    Consequently, the described reduction is tight.
\end{proof}

\fi

Next, we show that also the closely related \textsc{Sparsest Cut} is \PLS{}-complete under the \textsc{Flip} neighborhood.
Note that \textsc{Densest Cut} and \textsc{Sparsest Cut} are both \NP-hard~\cite{MatulaS90}.
\textsc{Sparsest Cut} is studied intensively in terms of approximation algorithms~\cite{AroraRV09} and integrality 
gaps~\cite{KaneM13}, and is used to reveal the hierarchical community structure of social networks~\cite{MannMO08}
and in image segmentation~\cite{ShiM00}.

\begin{corollary}\label{cor:sparsest-cut}
    There exists a tight \PLS{}-reduction from \densestcut{}  to \sparsestcut{}.
\end{corollary}

\iflong

\begin{proof}
    Let~$G=(V,E)$ with weight function~$w$ be an instance of \densestcut{}.
    Without loss of generality we assume that~$G$ is a complete graph by setting~$w(uv)\coloneqq 0$ for each non-edge of~$G$.
    Let~$M$ be the maximum weight of any edge.
    The instance of \sparsestcut{} is constructed as follows: 
    The graph~$G'$ has the same vertex set~$V$ as~$G$, and we 
    set~$w'(uv)\coloneqq  M -w(uv)$ for each edge~$uv\in E$.
    
    Observe that for a partition~$(X,Y)$ of the vertices of~$G'$ we obtain that 
    $$\frac{w'(X,Y)}{|X|\cdot |Y|}=\frac{M\cdot |X|\cdot |Y|-w(X,Y)}{|X|\cdot |Y|}= M -\frac{w(X,Y)}{|X|\cdot |Y|}.$$
    Hence, for any two partitions~$(X,Y)$ and~$(A,B)$, $\frac{w'(X,Y)}{|X|\cdot |Y|} < \frac{w'(A,B)}{|A|\cdot |B|}$ if and only if~$\frac{w(X,Y)}{|X|\cdot |Y|} > \frac{w(A,B)}{|A|\cdot |B|}$.
    Thus, both problem instances share the same locally optimal solutions, which proves the correctness of the reduction.
    Tightness follows by observing that the two instances have the same transition graphs.
\end{proof}

\fi

\iflong 
We are now ready to prove the second step, 
\fi
that \hartigan{} is \PLS{}-\iflong complete\fi. 
\iflong
Note that in
the following reductions, as well as in
\Cref{lemma:bisection_to_emaxcut,lemma:bisection_to_emaxcut_nosquare}, we construct
sets of points $\X \subseteq\mathds{R}^d$ as instances to the clustering problems. This makes
the reductions easier to reason about. The point sets we construct have the property
that for each $x, y\in \X$, it holds that $\|x - y\|^2 \in \mathds{Z}$ (or $\|x - y\| \in \mathds{Z}$
in the case of \Cref{lemma:bisection_to_emaxcut_nosquare}). Thus, these point sets
represent valid instances of each clustering problem.
\fi

\iflong

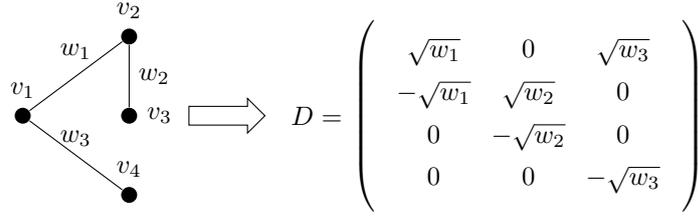
\begin{figure}[t]
\centering
\begin{tikzpicture}[scale=0.7]
    \vertex[label=$v_1$](v1) at (-0.5, 0) {};
    \vertex[label=$v_2$](v2) at (1.5, 1.5) {};
    \vertex[label=right:$v_3$](v3) at (1.5, 0) {};
    \vertex[label=$v_4$](v4) at (1.5, -1.5) {};

    \draw (v1) -- node[above=3pt] {$w_1$} ++ (v2);
    \draw (v2) -- node[right] {$w_2$} ++ (v3);
    \draw (v1) -- node [above] {$w_3$} ++ (v4);

    \coordinate(x) at (2.5, 0) {};
    \coordinate(y) at (3.5, 0) {};

    \node (arrow) [single arrow, draw=black,
      minimum width = 10pt, single arrow head extend=3pt,
      minimum height=10mm] at (3.25, 0) {};

    \node (D) [right=6pt of arrow] {$D=$};
    \node (matrix) [matrix of math nodes,
                right=10pt of D,
                left delimiter = (,right delimiter = ),
                row sep=-1pt,
                column sep = 1pt] {
                    \sqrt{w_1} & 0 & \sqrt{w_3} \\
                    -\sqrt{w_1} & \sqrt{w_2} & 0 \\
                    0 & -\sqrt{w_2} & 0 \\
                    0 & 0 & -\sqrt{w_3} \\
                };
\end{tikzpicture}

\caption{A simple \textsc{Densest Cut} instance.
The matrix $D$ on the right is constructed from the graph on the left in
the proof of \Cref{lemma:densestcut_to_hartigan}. The rows of $D$ correspond to points
in an instance of \hartigan{}. 
The actual instance of \textsc{2-Means} does not use these
points directly, as we only have to encode the weights, which are squared
distances and integral in this reduction.}
\label{fig:hartigan}
\end{figure}

\fi

\begin{lemma}\label{lemma:densestcut_to_hartigan}
    There exists a tight \PLS{}-reduction from \densestcut{} without isolated vertices to \hartigan{}.
\end{lemma}

\iflong

\begin{proof}
    Our reduction is based on an \NP-hardness reduction from \textsc{Densest Cut} to \textsc{2-Means}~\cite{ADHP09}.
    Our construction is basically the same as the one used by Aloise et al.~\cite{ADHP09}. 
    The only difference is that they reduced from the unweighted variant of \textsc{Densest Cut}.
    Since we want to show a \PLS{}-reduction, we have to incorporate edge weights.
    
    Let~$G=(V,E)$ be an instance of \densestcut{} where each vertex is incident with at least one edge.
    We construct an equivalent instance of \hartigan{} as follows:
    For each vertex~$v\in V$ we add a point in~$\mathds{R}^{|E|}$.
    Abusing notation, we may also denote this corresponding point by~$v$.
    In other words, all points together correspond to a $|V|\times |E|$~matrix~$M$.
    For each edge~$e=xy$ the entries of~$M$ are defined as follows:
    For each vertex~$v\not\in e$, we set~$M(v,e)\coloneqq 0$, and we set~$M(x,e)\coloneqq \sqrt{w(e)}$ and~$M(y,e)\coloneqq -\sqrt{w(e)}$.
    Note that it is not important which of the entries~$M(x,e)$ and~$M(y,e)$ has positive sign and which has negative sign; it is only important that they are different.
    Also, note that the only distinction in our construction to the one used by Aloise et al.~\cite{ADHP09} is that we set the 
    matrix entries to~$\pm\sqrt{w(e)}$. In contrast, Aloise et al. set~\cite{ADHP09} the entries to~$\pm 1$ (see \Cref{fig:hartigan}).
    
    \proofsubparagraph{Correctness.}
    We exploit the following observation.
    Let~$(Q,R)$ be a clustering of~$V$ and let~$q=|Q|$, and~$r=|R|$.
    The coordinates of the centroids~$c_Q$ of~$Q$ and~$c_R$ of~$R$ can be computed as the mean of all points in 
    the respective cluster as follows (by~$c_Q(e)$ and~$c_R(e)$ we denote the $e^\text{th}$
    coordinate of~$c_Q$ and~$c_R$ respectively):
    If~$e\in E(Q,R)$, then~$c_Q(e)= \pm\frac{\sqrt{w(e)}}{q}$ 
    and~$c_R(e)= \mp\frac{\sqrt{w(e)}}{r}$.
    Otherwise, if~$e\not\in E(P,Q)$, then~$c_Q(e) = 0$ and also~$c_R(e)= 0$.
    
    Based on this observation, we can now analyze the structure of locally optimal solutions for the constructed instance of~\hartigan{}.
    First, we show that each clustering~$(Q,R)$ of~$V$ with~$Q = \emptyset$ can be improved by flipping any point from~$R$ to~$Q$.
    Note that by the above, this implies that~$E(Q,R) = \emptyset$ and thus~$c_Q =  c_R = 0$.
    Let~$v$ be an arbitrary point of~$R$.
    Since we assumed that each vertex in~$G$ is incident with at least one edge, $M(v,e) \neq 0$ for some edge~$e\in E$.
    Hence, $\|v-c_R\|^2 > 0$.
We set~$(Q' \coloneqq \{v\}, R' \coloneqq R\setminus \{v\})$.
Since~$|Q'|=1$, we have~$v=c_{Q'}$ and thus~$\|v-c_{Q'}\| = 0$.
Furthermore, since~$c_{R'}$ is the centroid of~$R'\subseteq R$ and since~$\|v-c_R\|^2 > 0$ we have~$\sum_{u\in R'} \|u - c_{R'}\|^2 \leq \sum_{u\in R'} \|u - c_{R}\|^2$.
    
    Hence, in the following, we assume that for the clustering~$(Q,R)$, both~$Q$ and~$R$ are non-empty.
    Then, the value of a clustering~$(Q,R)$ can be computed as follows:
    
    \begin{align*}
        \sum_{e\in E} &\left(\text{cost of~$R$ due to the $e^{\text{th}}$~coordinate + cost of~$Q$ due to the $e$-th~coordinate}\right) \\
        &= \sum_{\substack{e=uv\in E(Q,R)\\u\in Q,v\in R}} \ \Biggl(\sum_{x\in R\setminus \{v\}}(x(e)-c_R(e))^2+(v(e)-c_R(e))^2 \\
        &~~~~~~+ \sum_{x\in Q\setminus \{u\}}(x(e)-c_Q(e))^2+(u(e)-c_Q(e))^2\Biggr) + \sum_{e\not\in E(Q,R)} 2\sqrt{w(e)}^2\\
        &= \sum_{e\in E(Q,R)} (r-1)\left(\frac{\sqrt{w(e)}}{r}\right)^2+\left(\sqrt{w(e)}-\frac{\sqrt{w(e)}}{r}\right)^2 \\
        &~~~~~~+ (q-1)\left(\frac{\sqrt{w(e)}}{q}\right)^2+\left(\sqrt{w(e)}-\frac{\sqrt{w(e)}}{q}\right)^2 
            + \sum_{e\not\in E(Q,R)} 2\sqrt{w(e)}^2.
    \end{align*}
    
    This equation can be simplified to:
    \begin{align*}
        \sum_{e\in E(Q,R)} w(e)&\left[\frac{r-1}{r^2}+\left(1-\frac{1}{r}\right)^2 + \frac{q-1}{q^2}+\left(1-\frac{1}{q}\right)^2\right] + \sum_{e\not\in E(Q,R)} 2 w(e)\\
            &= \sum_{e\in E(Q,R)} w(e)\left[1-\frac{1}{r} + 1-\frac{1}{q}\right] + \sum_{e\not\in E(Q,R)} 2 w(e)\\
            &= \left(2-\frac{1}{p}-\frac{1}{q}\right) w(Q,R) + 2\left[w(Q,Q)+w(R,R)\right]
            = w(E)-n\cdot\frac{w(Q,R)}{q\cdot r}.
    \end{align*}
    
    Here, the last equality follows from the fact that~$r+q=n$ since~$(Q,R)$ is a clustering of~$V$.
    Thus, a clustering~$(Q',R')$ of~$V$ where both~$Q'$ and~$R'$ are non-empty improves over~$(Q,R)$ if and only if~$\frac{w(Q,R)}{q\cdot r} > \frac{w(Q',R')}{|Q'|\cdot |R'|}$.
    That is, if and only if~$(Q',R')$ is a denser cut of~$G$ than~$(Q,R)$.
    Hence, both problem instances share the same locally optimal solutions, which proves the correctness of the reduction.
    
    \proofsubparagraph{Tightness.}
    Let $\R$ be the set of all solutions such that neither cluster
    is empty. By the arguments above, we know that $\R$ contains all local optima, and constructing
    a clustering in~$\R$ from a solution to the original \textsc{Densest Cut} instance is trivial. This gives
    us the first two properties of tight reductions.
    Moreover, we know that moving from some solution $s \in \R$ to $s' \notin \R$ makes the clustering strictly worse.
    Thus, property 3 of tight reductions follows by observing that the transition graph of the reduced instance
    restricted to $\R$ is exactly the transition graph of the original instance.
\end{proof}

\fi

Finally, we provide a generic reduction to show \PLS{}-\iflong completeness \fi for general~$k$. \iflong\else
\fi

\begin{lemma}\label{lemma:k-means}
    For each~$k\geq 2$, there exists a tight \PLS{}-reduction from \khartigan{} to \kkhartigan{}.
\end{lemma}

\iflong

\begin{proof}
Let~$\X\subseteq \mathds{R}^d$ be an instance of \khartigan{} consisting of $n$~points. 
By~$D$ we denote the sum of all squared distances between each two points of~$\X$.
We construct an instance~$\X'\subseteq \mathds{R}^d$ of \kkhartigan{} consisting of $(n+1)$~points as follows: 
$\X'$ contains a copy of~$\X$ and one additional point~$z$.
The coordinates of~$z$ are set to the coordinates of an arbitrary but fixed point~$x\in\X$ plus an offset of~$3n\cdot D$ in each dimension.

\proofsubparagraph{Correctness.}
First, observe that it is safe to assume that each of the $(k+1)$~clusters of a solution of~$\X'$ is non-empty:
Let~$C$ be any cluster containing at least two points. 
Since no two points share the same coordinates, at least one point~$p\in C$ has a positive distance from the centroid of~$C$.
Now, we create a new cluster consisting of point~$p$.
Note that this new clustering yields a better solution and thus the assumption is justified.

Second, we show that a clustering is not locally optimal if~$z$ is in a cluster~$C$ with at least one other point~$x\in\X$.
More precisely, we show that moving~$x$ from~$C$ into~$C'$, where~$C'$ is any other cluster, is improving.
Since the centroid of~$C$ has distance at least~$3D$ to~$x$ in each dimension, the contribution of the new cluster~$C\setminus\{x\}$ to the clustering is at least~$2D$ lower than the contribution of cluster~$C$.
Furthermore, since the sum of all distances of points in~$\X$ is~$D$, the contribution to the clustering of~$C'\cup\{x\}$ increases by at most~$D$.
Hence, moving~$x$ from~$C$ to~$C'$ is improving.

Thus, we can safely assume that one cluster only contains~$z$.
It follows that both instances share the same locally optimal solution.
For tightness, we set $\R$ equal to the clusterings such that $z$ is in a cluster by itself.
By the arguments above, $\R$ contains all local optima, and moving from some clustering in $\R$
to a clustering outside of $\R$ strictly worsens the clustering. Thus, tightness follows since the transition graph restricted to
$\R$ is equal to the transition graph of the original instance.
\end{proof}

\fi
Now, \Cref{thm:hartigan} follows by applying the tight \PLS{}-reductions according to 
\Cref{fig:reduction_graph}.

\iflong
\subsection{Squared Euclidean Max Cut}
\fi

We construct a \PLS-reduction from \minbisection{} to \emaxcut{}. 
The reduction is largely
based on the \NP-hardness proof of \textsc{Euclidean Max Cut}
of Ageev et al.\
\cite{ageevComplexityWeightedMaxcut2014}. The main difference is that we must
incorporate the weights of the edges of the \minbisection{} instance into the
reduction. 
\iflong Note the similarity of this reduction to the one
used in \Cref{lemma:densestcut_to_hartigan}.
\fi

\begin{lemma}\label{lemma:bisection_to_emaxcut}
    There exists a tight \PLS{}-reduction from \minbisection{} to \emaxcut{}.
\end{lemma}

\iflong

\begin{proof}
    Let $G = (V, E)$ be an instance
    of \minbisection{} with weights $w$ on the edges,
    with $n = |V|$ and $m = |E|$. We assume that
    there exists at least one edge with nonzero weight. This is not a
    restriction, as the reduction is trivial when all weights are zero. 
    
    Let $M$ be the incidence matrix of $G$,
    modified such that each nonzero entry is replaced by the
    square root of half of its corresponding
    edge weight, and let~$J$ be an $n \times n$ diagonal matrix with rows
    indexed by the vertices of $G$. The entries of $J$ are set equal 
    to~$\alpha_v = \sqrt{w(E)/2 - w(\delta(v))/2}$.
    We then construct a matrix~$D = (M \mid J)$.
    Note that each row of $D$ corresponds to a vertex of $G$.
    
    From this matrix, we construct an instance
    of \emaxcut{}. For each row of $D$, we create a point
    in $\mathds{R}^{m+n}$ with coordinate vector equal to this row.
    In this way, the points are identified with the vertices $V(G)$.
    We call the set formed by these points $\X$.

\proofsubparagraph{Correctness.}
    Let $(X, Y)$ be a solution to this instance of
    \textsc{Squared Euclidean Max Cut}. If $|\,|X| - |Y|\,| = 1$, we map $(X, Y)$ to a solution of
    \textsc{Min Bisection} using the same partition. Otherwise,
    we map $(X, Y)$ to any solution of \textsc{Min Bisection} arbitrarily; we will
    show that these solutions cannot be locally optimal, and thus we may
    ignore them.
    
    From the definition of $D$, it follows that that every point $x \in \X$ has $\|x\|^2 = w(E)/2$.
    Moreover, for two distinct points $x, y \in \X$, their dot product
    $\langle x, y \rangle$ is nonzero if and only if~$x$ and $y$ are adjacent, in which case
    it evaluates to $w(xy)/2$. From these observations, the
    cost of a solution $(X, Y)$ is
    \begin{align*}
         \sum_{x \in X} \sum_{y \in Y} \|x - y\|^2
            &= \sum_{x \in X} \sum_{y \in Y} \left(\|x\|^2 + \|y\|^2
                    - 2\langle x, y\rangle\right)\\
            &= |X|\cdot |Y| \cdot w(E) - w(X, Y).
    \end{align*}

     We consider the first term in the expression above.
     This term is maximized when $|\,|X| - |Y|\,| = 1$; as this difference increases,
     $|X|\cdot |Y|$ decreases monotonically in integer steps of magnitude at least 2.
     Therefore, any flip that increases $|\,|X| - |Y|\,|$ decreases
     the cost by at least $2\cdot w(E)$, and increases it by
     at most $w(E)$, for a strictly negative contribution (since we assume $w(E) > 0$).

     By this argument, no solution $(X, Y)$ with $|\,|X| - |Y|\,| > 1$
     can be locally optimal under the \textsc{Flip} neighborhood.
     Thus, we restrict our attention to solutions such
     that $|\,|X| - |Y|\,| = 1$. For these solutions, $|X|\cdot|Y|$ is a fixed term,
     and so the only relevant part of the cost is the term
     $-w(X, Y)$.
     
     It is now easy to see that any
     locally maximal solution for this instance
     of \emaxcut{} maps to a locally minimal solution
     of \minbisection{}. 
     
     \proofsubparagraph{Tightness.} We let~$\R$ be the set of solutions with~$|\,|X| - |Y|\,| = 1$.
     We already know that $\R$ contains all local optima, and constructing some solution in $\R$
     from a solution to the original \minbisection{} instance is trivial.
     By the arguments above, we also see that moving from a cut in $\R$ to a cut
     outside $\R$ strictly decreases the cut value. Tightness now follows
     by observing that the transition graph of the reduced instance restricted to $\R$ is identical
     to the transition graph of the original instance.
\end{proof}

\fi

With a few modifications, the proof can be adapted to a reduction
to \emaxcutnosq{}. The main challenge in adapting the proof is that
the objective function is now of the form $\sum \|x - y\|$, rather than
$\sum \|x - y\|^2$. However, by suitably modifying the coordinates of the points,
the distances $\|x - y\|$ in the \textsc{Euclidean Max Cut} instance can take the same value
as $\|x - y\|^2$ in the \textsc{Squared Euclidean Max Cut} instance.

\begin{lemma}\label{lemma:bisection_to_emaxcut_nosquare}
    There exists a tight \PLS{}-reduction from \minbisection{} to \emaxcutnosq{}.
\end{lemma}

\iflong

\begin{proof}
    The reduction is identical to the reduction in the proof of \Cref{lemma:bisection_to_emaxcut}, except
    that the incidence matrix entries become $\sqrt{C\cdot w(e)\cdot w(E) - w(e)^2/2}$ and
    the diagonal matrix entries become 
    \[
        \alpha_v = \sqrt{C^2\cdot w(E)^2/2 - \sum_{e \in \delta(v)}\left(
            C\cdot w(e)\cdot w(E) - w(e)^2/2
        \right)}
    \]
    for some constant $C \geq 2$. For this to work, the quantities under the radicals must be non-negative.
    Clearly, $C\cdot w(e)\cdot w(E) - w(e)^2/2 \geq 0$, so the incidence matrix
    entries are valid. 
     For the numbers $\alpha_v$, it suffices to show that
    \[
        \sum_{e \in \delta(v)} \left(2\cdot C\cdot w(e)\cdot w(E) - w(e)^2\right) \leq 2\cdot C\cdot w(E)^2 - \sum_{e \in \delta(v)} w(e)^2 \leq C^2\cdot w(E)^2.
    \]
    Since $C \geq 2$, this is satisfied.
    
    \proofsubparagraph{Correctness.}    
    Similarly to the proof of \Cref{lemma:bisection_to_emaxcut}, we have for each
    $x \in \X$ that $\|x\|^2 = C^2\cdot w(E)^2/2$. Meanwhile, for $x, y \in \X$,
    the inner product $\langle x, y\rangle$ evaluates to
    $C\cdot w(xy)\cdot w(E) - w(xy)^2/2$ if $x$ and $y$ are adjacent, and zero otherwise. We find
    \begin{align*}
        \|x-y\| &= \sqrt{\|x\|^2 + \|y\|^2 - 2\langle x, y\rangle}\\
            &= \sqrt{C^2\cdot w(E)^2 + w(xy)^2 - 2\cdot C\cdot w(E)\cdot w(xy)} = C\cdot w(E) - w(xy).
    \end{align*}
    Thus, the cost function of the reduced instance is $C\cdot |X|\cdot|Y|\cdot w(E) - w(X, Y)$.
    The remainder of the proof, including the proof that the reduction is tight,
    is identical to that of \Cref{lemma:bisection_to_emaxcut}.
\end{proof}

This completes the proof of \Cref{thm:maxcut}.

\fi

\section{Discussion}

\Cref{thm:hartigan,thm:hardinstances} show that no local improvement algorithm
using the \textsc{Flip} heuristic can find locally optimal clusterings efficiently,
even when $k=2$.
This result augments an earlier worst-case
construction \cite{mantheyWorstCaseSmoothedAnalysis2024b}.
\Cref{thm:maxcut} demonstrates that finding local optima in \textsc{Squared Euclidean Max Cut}
is no easier than for general \textsc{Max Cut} under the \textsc{Flip} neighborhood. Thus,
the Euclidean structure of the problem yields no benefits with respect to the
computational complexity of local optimization. 

\subparagraph{Smoothed Analysis.}
Other \PLS{}-hard problems have yielded under
smoothed analysis. Chiefly, \maxcut{} has polynomial
smoothed complexity in complete graphs \cite{angelLocalMaxcutSmoothed2017,bibakImprovingSmoothedComplexity2021}
and quasi-polynomial smoothed
complexity in general graphs \cite{chenSmoothedComplexityLocal2020a,etscheidSmoothedAnalysisLocal2017}.
We hope that our results here serve to motivate
research into the smoothed complexity of \khartigan{} and \emaxcut{}, with the goal of adding them to the list of hard local search problems
that become easy under perturbations.

\subparagraph{Reducing the Dimensionality.} Our reductions yield instances of
\textsc{$k$-Means} and \textsc{(Squared) Euclidean Max Cut} in
$\Omega(n)$ dimensions.
Seeing as our reductions cannot be obviously adapted for $d = o(n)$,
we raise the question of whether the hardness of \emaxcut{} and \khartigan{} is
preserved for $d = o(n)$. This seems unlikely for
\emaxcut{} for $d = O(1)$,
since there exists an $O(n^{d+1})$-time exact algorithm
due to Schulman \cite{schulmanClusteringEdgecostMinimization2000}.
A direct consequence of \PLS{}-hardness for $d = f(n)$ would thus be an $O\left(n^{f(n)}\right)$-time
general-purpose local optimization algorithm. Concretely, \PLS{}-hardness for
$d = \polylog n$ would yield a quasi-polynomial time algorithm for all problems in \PLS{}.

For \khartigan{}, the situation is similar: For $d = 1$,
\textsc{$k$-Means} is polynomial-time solvable for any $k$.
However, already for $d = 2$,
the problem is \NP-hard \cite{mahajanPlanarKmeansProblem2012} when~$k$ is arbitrary.
When both $k$ and $d$ are constants, the problem is again
polynomial-time solvable,
as an algorithm exists that finds an optimal
clustering in time $n^{O(kd)}$ \cite{hasegawaEfficientAlgorithmsVarianceBased2000}.
Thus, \PLS{}-hardness for $kd \in O(f(n))$ would yield an $n^{O(f(n))}$-time
algorithm for all \PLS{} problems in this case.

\subparagraph{Euclidean Local Search.} There appear to be very few \PLS{}-hardness results for Euclidean local
optimization problems, barring the result of Brauer \cite{brauerComplexitySingleSwapHeuristics2017}
and now \Cref{thm:hartigan} and \Cref{thm:maxcut}. A major challenge in obtaining
such results is that Euclidean space
is very restrictive; edge weights cannot be independently set, so the intricacy often required
for \PLS-reductions is hard to achieve. Even in the present work, most of the work is done
in a purely combinatorial setting. It is then useful to get rid of the Euclidean structure
of the problem as quickly as possible, which we achieved by modifying the reductions of
Ageev et al.\ \cite{ageevComplexityWeightedMaxcut2014} and Alois et al.\ \cite{ADHP09}.

With this insight, we pose the question of what other local search problems remain
\PLS{}-hard for Euclidean instances. Specifically, is \textsc{TSP}
with the $k$-opt neighborhood still \PLS{}-hard in Euclidean (or squared Euclidean)
instances, for sufficiently large $k$? This is known to be the case for general
metric instances for some large constant $k$ \cite{krentelStructureLocallyOptimal1989a} (recently
improved to all $k \geq 17$ \cite{heimannOptAlgorithmTraveling2024}), but Euclidean instances
are still a good deal more restricted.

\bibliographystyle{plain}
\bibliography{bibliography}

\end{document}